\documentclass{amsart}

\usepackage{amscd}
\usepackage{bbm}
\usepackage{mathrsfs}
\usepackage{eucal}
\usepackage{faktor}
\usepackage{xfrac}
\usepackage{braket}
\usepackage{relsize} 
\usepackage{nccmath}
\usepackage{mathtools}
\usepackage{comment}
\usepackage{verbatim}
\usepackage{wasysym}
\usepackage{graphicx}
\usepackage{caption}
\usepackage{import}
\usepackage{latexsym}
\usepackage{amsfonts}
\usepackage{amssymb}
\usepackage{amsmath}
\usepackage{mathtext}
\usepackage{cite}
\usepackage{enumerate}
\usepackage{float}
\usepackage{amsthm}
\usepackage{upgreek}
\usepackage{nicefrac}
\usepackage{multicol}
\usepackage{rotating}
\usepackage{array}

\usepackage{etoolbox}
\makeatletter
\pretocmd{\chapter}{\addtocontents{toc}{\protect\addvspace{15\p@}}}{}{}
\pretocmd{\section}{\addtocontents{toc}{\protect\addvspace{5\p@}}}{}{}
\makeatother
\let\oldtocsection=\tocsection
\let\oldtocsubsection=\tocsubsection
\let\oldtocsubsubsection=\tocsubsubsection
\renewcommand{\tocsection}[2]{\hspace{0em}\oldtocsection{#1}{#2}}
\renewcommand{\tocsubsection}[2]{\hspace{1.8em}\oldtocsubsection{#1}{#2}}
\renewcommand{\tocsubsubsection}[2]{\hspace{4.4em}\oldtocsubsubsection{#1}{#2}}

\usepackage[svgnames]{xcolor}
\definecolor{linkcolor}{HTML}{e88d67} 
\definecolor{citecolor}{HTML}{e88d67} 
\definecolor{urlcolor}{HTML}{e88d67} 

\usepackage[
  colorlinks = , 
  urlcolor = urlcolor, 
  linkcolor = linkcolor, 
  citecolor = citecolor, 
  linktoc = all,
  pagebackref = true
]{hyperref} 

\mathtoolsset{showonlyrefs,showmanualtags}

\DeclareMathOperator{\vsp}{\phantom{\dfrac{\dfrac{}{}}{\dfrac{}{}}}}

\DeclareMathOperator{\ad}{ad}
\newcommand{\Ai}[1]{\text{Ai} \left( #1 \right)}
\newcommand{\brackets}[1]{\left( #1 \right)}
\newcommand{\PoissonBrackets}[1]{\left\{ #1 \right\}}
\newcommand{\LieBrackets}[1]{\left[ #1 \right]}
\newcommand{\angleBrackets}[1]{\langle #1 \rangle}
\newcommand{\abs}[1]{\left| #1 \right|}

\newcommand{\typeA}[1]{A_{#1}^{\brackets{1}}}
\newcommand{\PIIn}[1]{\text{P}_\text{II}^{\brackets{ #1 }}}
\newcommand{\HIIn}[1]{\text{H}_\text{II}^{\brackets{ #1 }}}

\newcommand{\dsum}{\displaystyle \sum}
\newcommand{\quot}[2]{{\raisebox{.3em}{$#1\!$}\bigl/\raisebox{-.3em}{$\!#2$}}}
\newcommand{\Painleve}{Painlev{\'e} }
\newcommand{\Backlund}{B{\"a}cklund }
\newcommand{\Mobius}{M{\"o}bius }

\theoremstyle{plain}
\newtheorem{thm}{Theorem}[]

\newtheorem{lem}{Lemma}[]

\newtheorem{prop}{Proposition}[]

\newtheorem{corl}{Corollary}[]

\theoremstyle{definition}
\newtheorem{defn}{Definition}[]

\newtheorem{exmp}{Example}[]

\theoremstyle{remark}
\newtheorem{rem}{Remark}

\usepackage{chngcntr}
\counterwithin*{thm}{section} 
\counterwithin*{lem}{section} 
\counterwithin*{claim}{section} 
\counterwithin*{prop}{section} 
\counterwithin*{corl}{section} 
\counterwithin*{defn}{section} 
\counterwithin*{exmp}{section} 
\counterwithin*{ex}{section} 
\counterwithin*{rem}{section} 

\usepackage{apptools}
\AtAppendix{\counterwithin{thm}{section}}
\AtAppendix{\counterwithin{lem}{section}}
\AtAppendix{\counterwithin{claim}{section}}
\AtAppendix{\counterwithin{prop}{section}}
\AtAppendix{\counterwithin{corl}{section}}
\AtAppendix{\counterwithin{defn}{section}}
\AtAppendix{\counterwithin{exmp}{section}}
\AtAppendix{\counterwithin{ex}{section}}
\AtAppendix{\counterwithin{rem}{section}}

\usepackage{tikz}

\begin{document}

    \title[On symmetries of the non-stationary $\PIIn{n}$ hierarchy and~their~applications]{On symmetries of the non-stationary $\PIIn{n}$ hierarchy and~their~applications}

    \author{Irina Bobrova}
    \noindent\address{\noindent Faculty of Mathematics, National Research University Higher School of Economics, Usacheva str. 6, Moscow, Russia}
    \email{ia.bobrova94@gmail.com}


    \maketitle

\begin{abstract}
In the current paper we study auto-\Backlund transformations of the non-stationary second \Painleve hierarchy $\PIIn{n}$ depending on $n$ parameters: a parameter $\alpha_n$ and \textit{times} $t_1, \dots, t_{n-1}$. Using generators $s^{(n)}$ and $r^{(n)}$ of these symmetries, we have constructed an affine Weyl group $W^{(n)}$ and its extension $\Tilde{W}^{(n)}$ associated with the $n$-th member of considered hierarchy. We determined $\PIIn{n}$ rational solutions via Yablonskii-Vorobiev-type polynomials $u_m^{(n)} \brackets{z}$. We brought out a relation between Yablonskii-Vorobiev-type polynomials and polynomial $\tau$-functions $\tau_m^{(n)} \brackets{z}$ and found their determinant representation in the Jacobi-Trudi form.
\end{abstract}

\tableofcontents

\section*{Introduction}

The \Painleve equations appeared as a result of a classification of complex second-order differential equations with no movable singular points except of poles
\begin{equation}
    \dfrac{d^2 w \brackets{z}}{d z^2}
    = P \brackets{z, w \brackets{z}, \dfrac{d w \brackets{z}}{dz}},
\end{equation}
where $P \brackets{z, w \brackets{z},  \left. d w \brackets{z}\right/dz}$ is a meromorphic function in $z$ and is a rational function in $w \brackets{z}$ and its derivative. 50 classes were found by P. \Painleve and his school \cite{painleve1900memoire, painleve1902equations}. Later B. Gambier \cite{gambier1910equations} proved that only six of them define new special functions which are called \textit{the \Painleve transcendent}. These six equations are known as \textit{the \Painleve equations}.

The constant interest in the \Painleve equations is due to their mathematical beauty in various fields of mathematics (for instance, each of the \Painleve equations arises as an isomonodromic problem on the Riemann sphere) and their frequent use in mathematical physics, especially in the applications to classical, quantum, and discrete integrable systems. Although the \Painleve equations were introduced more than 100 years ago, some areas of its applications are still poorly understood. 

Commonly known that auto-\Backlund transformations of the second \Painleve equation in the following form
\begin{equation} \label{PII1}
    \PIIn{1}
    \LieBrackets{w \brackets{z;\alpha_1}}
    :
    \quad
    \dfrac{d^2 w}{dz^2}
    = 2 w^3 
    + z w 
    + \alpha_1,
    \quad
    \alpha_1 \in \mathbb{C},
\end{equation}
have a group structure related to an affine Weyl group $W^{(1)} = \angleBrackets{s^{(1)}, r^{(1)}}$ of type $\typeA{1}$, i.e. the fundamental relations are $\brackets{s^{(1)}}^2 = 1$ and $\brackets{r^{(1)}}^2 = 1$. On the one hand, this fact is helpful to generating new solutions of the second \Painleve equation from a trivial one. In other words, using \Backlund transformations, we can construct the dynamics on the space of initial data \cite{okamoto1986studies}, thereby defining the discrete dynamics on this space \cite{sakai2001rational}. On the other hand, the second \Painleve equation has a symmetric form in some coordinates obtained by an extended affine Weyl group $\tilde{W}^{(1)} = \angleBrackets{s^{(1)}_0, s^{(1)}_1; \pi^{(1)}}$, where $s^{(1)}_0 = s^{(1)}$, $s^{(1)}_1 = r^{(1)} s^{(1)} r^{(1)}$, $\pi^{(1)} = r^{(1)}$, with an additional fundamental relation $\pi^{(1)} s^{(1)}_{i} = s^{(1)}_{i + 1} \pi^{(1)}$, $i \in \quot{\mathbb{Z}}{2 \mathbb{Z}}$. 

It is also well-known that the second \Painleve equation is a result of a self-similar reduction from the modified Korteveg-de Vries equation \cite{kudryashov1997first}. Since the latter has the $\tau$-function, the second \Painleve equation also has it. One of the amazing fact about the $\tau$-function in the \Painleve theory is that for the second \Painleve equation it has two determinant representations: in the Jacobi-Trudi form and in the Hankel form \cite{kajiwara1996determinant}. The $\tau$-function of the second \Painleve equation also appears as a solution of the bilinear system \cite{fukutani2000special} and thus is determined on a corresponding lattice. One of the equations of this system is known as o \textit{the Toda lattice equation} \cite{kajiwara2007remark}. This fact arises from a symmetric form constructed by an extended affine Weyl group \cite{noumi2004painleve}.

Before papers \cite{yablonskii1959rational} and \cite{vorobiev1965rational}, it was believed that the second \Painleve equation determines only new special functions. However, A. Yablonskii and A. Vorobiev proved that for integer parameters one can construct rational solutions by special polynomials called \textit{the Yablonskii-Vorobiev polynomials}. This phenomena appears for all \Painleve equations. So, for example, rational solutions of the third \Painleve equation are connected with \textit{the Umemura special polynomials} \cite{clarkson2006special}. Thus, one of the classical areas of the \Painleve equations theory is special polynomials and related to them topics (for instance, determinant representations, root plots on the complex plane) \cite{clarkson2003second,bertola2015zeros}. In addition to special polynomials, some \Painleve equations are associated with known special functions. For instance, for a half-integer parameter $\alpha_1$, the second \Painleve equation has a solution in terms of the Airy functions.

Since the \Painleve equations and related to them hierarchies are closely connected with the well-known integrable PDEs and their higher analogs, the investigation of them is so significant today for mathematical physics. In particular, the reason behind the omnipresent appearance of these equations is that they are innately linked to the Toda hierarchy.

There are several different generalizations of the \Painleve equations. In this paper, we study symmetries of \textit{the non-stationary hierarchy of the second \Painleve equation} obtained in work \cite{mazzocco2007hamiltonian}:
\begin{gather} \label{nonstatPIIhier}
    \PIIn{n} \LieBrackets{w \brackets{z, \overline{t}; \alpha_n}}
    :
    \quad
    \brackets{\dfrac{d}{dz} + 2 w} 
    \sum_{l = 0}^{n} 
    t_l \mathcal{L}_l \LieBrackets{w^{\prime} - w^2}
    = \alpha_n,
    \quad 
    t_0 = - z,
    \,
    t_n = 1,
    \quad 
    \alpha_n \in \mathbb{C},
    \quad
    n \geq 1,
\end{gather}
where $\mathcal{L}_l \LieBrackets{u}$ are operators defined by the Lenard recursion relation
\begin{align}
    \partial_z {\mathcal{L}}_{l+1} 
    \LieBrackets{u} 
    &= \left( \partial_z^3 + 4 u \partial_z + 2 \partial_z u \right) 
    {\mathcal{L}}_l \LieBrackets{u},
    \\
    {\mathcal{L}}_0 \LieBrackets{u} 
    &= \dfrac{1}{2},
    \quad
    {\mathcal{L}}_l \LieBrackets{0} 
    = 0, 
    \,\,
    l > 0,
\end{align}
and $\overline{t} = \brackets{t_1, \dots, t_{n - 1}}$ is a set of parameters called \textit{times}.

\begin{rem}
   This hierarchy is sometimes called \textit{the generalized \text{\rm $\PIIn{n}$} hierarchy}.
\end{rem}

\begin{rem}[about notation]
   Further in this text, we will accept the following agreements:
   \begin{itemize}
       \item $\prime = \dfrac{d}{dz}$;
       
       \item the number of a member in \eqref{nonstatPIIhier} is fixed by the index $n$, which appears as a parameter index $\alpha_n$, or as the Lenard operator index $\mathcal{L}_n \LieBrackets{u}$, or as an upper index of an object about which we discuss (e.g. $w^{(n)} \brackets{z, \overline{t}}$).
   \end{itemize}
\end{rem}

\begin{exmp}[a few members of \eqref{nonstatPIIhier}] \label{exmpPII012}
    \begin{align}
        \tag{$n=0$}
        \mathcal{L}_0 = \dfrac{1}{2}, 
        &\quad 
        \left( z - 1 \right) w = - \alpha_0 \quad 
        \Rightarrow 
        \quad 
        w = - \dfrac{\alpha_0}{z-1};
        \\
        &\quad 
        \text{(not a differential equation!)}
        \\
        \tag{$n=1$} 
        \mathcal{L}_1 = w^{\prime} - w^2,
        &\quad 
        w^{\prime\prime} = 2 w^3 + z w + \alpha_1;
        \\
        \tag{$n=2$}
        \mathcal{L}_2 = \left( w^{\prime} - w^2 \right)^{\prime\prime} + 3 \left( w^{\prime} - w^2 \right)^{2}, 
        &\quad 
        w^{\prime\prime\prime\prime} - 10 w^2 w^{\prime\prime} - 10 w {w^{\prime}}^{2} + 6 w^5 + t_1 \brackets{w^{\prime\prime} - 2 w^3} = z w + \alpha_2. 
    \end{align}
\end{exmp}

The stationary case $\overline{t} = 0$ was obtained in the paper \cite{joshi2004second} by a self-similar reduction from the modified Korteveg-de Vriez hierarchy
\begin{gather}\label{hiermKdV}
    \partial_{T_{2n+1}} v + \left( \partial_{xx} + 2 \partial_x v \right) \ell_n	\left[ v_x - v^2 \right] = 0, 
\end{gather}
where an operator $ \ell_n	\left[u\right]$ satisfies the Lenard recurrence relation. The reduction is given by the following data
\begin{table}[H]
    \centering
    \begin{tabular}{c|c}
         The symmetry & Independent coordinates \\
         \hline
         \hline
         $X_n = x \partial_x + \brackets{2n + 1} \partial_{T_{2n + 1}} - v \partial_v$
         &
         $
         \begin{matrix}
         w \left( z \brackets{x, T_{2n + 1}} \right) =
         v \brackets{x, T_{2n+1}} \left( \left( 2n + 1 \right) T_{2n + 1} \right)^{\frac{1}{2n+1}},
         \vsp
         \\
         z \brackets{x, T_{2n + 1}} = \dfrac{x}{\left( \left( 2n + 1 \right) T_{2n + 1} \right)^{\frac{1}{2n+1}}}.
         \end{matrix}
         $
    \end{tabular}
\end{table}

To find the non-stationary case \eqref{nonstatPIIhier}, M. Mazzocco and M. Mo have used the fact that the modified Korteveg-de Vriez hierarchy \eqref{hiermKdV} has the Virasoro symmetries. Thus, one can consider the following reduction data
\begin{table}[H]
    \centering
    \begin{tabular}{c|c}
         The Virasoro symmetry & Independent coordinates \\
         \hline
         \hline
         $T_n = \dsum_{k = 0}^n \brackets{2 k + 1} T_{k + 1} \partial_{T_{k + 1}}$
         &
         $
         \begin{matrix}
         t_0 = - z,
         \vsp
         \\
         t_k = \dfrac{\brackets{2 k + 1} T_{k + 1}}{\brackets{\brackets{2 n + 1} T_{n + 1}}^{\frac{2 k + 1}{2 n + 1}}},
         \end{matrix}
         $
    \end{tabular}
\end{table}
\hspace{-0.55cm} by which M. Mazzocco and M. Mo obtained the non-stationary $\PIIn{n}$ hierarchy \eqref{nonstatPIIhier}.

\begin{rem}
The modified Korteveg-de Vriez hierarchy is a result of a \Backlund transformation called \textit{the Miura transformation} of the Korteveg-de Vriez hierarchy:
\begin{table}[H]
    \centering
    \begin{tabular}{ccc|c|cc}
         $\phantom{\dfrac{\dfrac{}{}}{}}$ &
         $\partial_{T_{2n + 1}} u + \partial_x \ell_{n + 1} \LieBrackets{u} = 0$
         &
         $\Rightarrow$
         &
         $ u = v_x - v^2 $
         &
         $\Rightarrow$
         &
         $ \partial_{T_{2n + 1}} v +  \brackets{\partial_{xx} + 2 \partial_x v} \ell_{n} \LieBrackets{v_x - v^2} = 0 $
         \\
         &
         \scriptsize{
         the KdV hierarchy}
         &
         &
         \scriptsize{
         the Miura transformation}
         &
         &
         \scriptsize{
         the mKdV hierarchy}
    \end{tabular}
\end{table}
\end{rem}

As we discussed above, the \Backlund transformations are a powerful tool in the analytical theory of differential equations and have many applications in the integrable systems theory. So, we are interested in two subjects of symmetries of \eqref{nonstatPIIhier}: auto-\Backlund transformations and related to them an affine Weyl group, and rational solutions constructed by Yabloskii-Vorobiev polynomials and their determinant representations. 

Auto-\Backlund transformations of the stationary $\PIIn{n}$ hierarchy were obtained in \cite{clarkson1999backlund}. The authors used a one-to-one correspondence between the stationary $\PIIn{n}$ hierarchy and the Schwarz hierarchy (also called Schwarz-P$_{\text{II}}$ hierarchy), which solutions are invariant under the \Mobius group action. We obtained a generalization of this result to the non-stationary case \eqref{nonstatPIIhier}. 
\begin{thm}
Auto-\Backlund transformations of the $n$-th member \eqref{nonstatPIIhier} are defined for $\alpha_n \neq \frac{1}{2}$ and are~given~as
\begin{align}
    \Tilde{w} \brackets{z, \overline{t}; \tilde{\alpha}_n}
    &= w \brackets{z, \overline{t}; \alpha_n} 
    - \dfrac{2 \alpha_n - \varepsilon}{2 \dsum_{l = 0}^{n} t_l \mathcal{L}_l \LieBrackets{\varepsilon w^{\prime} - w^2}},
    \quad 
    t_0 = -z,
    \,\, 
    t_n = 1,
    \quad
    \tilde \alpha_n = 1 - {\alpha}_n,
    \quad
    \varepsilon = \pm 1.
\end{align}
\end{thm}

\begin{rem}
   When $\alpha_n = \frac{1}{2}$, we have special integrals $I_{1/2}^{(n)}$,
   \begin{equation}
        I_{1/2}^{(n)}:
        \quad 
       \sum_{l = 0}^{n} t_l \mathcal{L}_l \LieBrackets{w^{\prime} - w^2}
       = 0,
   \end{equation}
    which satisfy the relation
    \begin{equation}
        \brackets{\dfrac{d}{dz} + 2 w} I_{1/2}^{(n)} = 0.
    \end{equation}
   $I_{1/2}^{(n)}$ can be expressed in terms of the Airy-type functions (an appropriate generalization of the Airy functions see in \cite{fujii2007higher}). Using \Backlund transformations above, we are able to define these special integrals for all half-integer parameters $\alpha_n \in \mathbb{Z} + \frac{1}{2}$.
\end{rem}

One can note that \Backlund transformations above have two generators
\begin{align}
    s^{(n)}:
    &
    \qquad 
     \Tilde{w} \brackets{z, \overline{t}; \tilde{\alpha}_n}
        = w \brackets{z, \overline{t}; \alpha_n} 
        - \dfrac{2 \alpha_n - 1}{2 \dsum_{l = 0}^{n} t_l \mathcal{L}_l \LieBrackets{w^{\prime} - w^2}},
    \\
    r^{(n)}:
    &
    \qquad 
     \Tilde{w} \brackets{z, \overline{t}; -{\alpha}_n}
        = - w\brackets{z, \overline{t}; {\alpha}_n}.
\end{align}
Their compositions $T_m^{(n)} = \brackets{r^{(n)} s^{(n)}}^m$ and $S_k^{(n)}  = \brackets{r^{(n)} s^{(n)}}^k s^{(n)}$ have a group structure associated with the $n$-th member of the non-stationary $\PIIn{n}$ hierarchy. In the case $n = 1$, this group structure is an affine Weyl group of type $\typeA{1}$ that acts on a parameter $\alpha_1$ by reflections with respect to zero and one-half. We generalized this fact to an arbitrary $n$ in \eqref{nonstatPIIhier} as follows.

Let $Q_i$, $P_i$, $i = 1, \dots, n$, be canonical coordinates for the $n$-th member of \eqref{nonstatPIIhier} defined in the work \cite{mazzocco2007hamiltonian}. By Theorem 6.1 in \cite{mazzocco2007hamiltonian}, these canonical coordinates are expressed as matrix entries $b_{2i+1}^{(n)}$ and $b_{2i}^{(n)}$ of a matrix $\mathcal{A}^{(n)}$ of isomonodromic deformation problem (15) in \cite{mazzocco2007hamiltonian}. Since canonical coordinates are expressed in terms of a solution $w \brackets{z, \overline{t}; \alpha_n}$ and its derivatives and \Backlund transformations act on them by rules $s^{(n)}$ and $r^{(n)}$, we found out an explicit form of generators of an affine Weyl group $W^{(n)}$ and its extension $\tilde{W}^{(n)}$ in terms of canonical coordinates $Q_i$, $P_i$, $i = 1, \dots, n$, for the $n$-th member of hierarchy \eqref{nonstatPIIhier}.

\begin{thm}
An affine Weyl group $W^{(n)} = \angleBrackets{s^{(n)}, r^{(n)}}$ associated with the $n$-th member \eqref{nonstatPIIhier} has a type $\typeA{1}$ and has generators $s^{(n)}$ and $r^{(n)}$ acting on canonical coordinates and a parameter $b_n = 2 \alpha_n - 1$ by the following rules
\begin{gather} 
    \tag{$s^{(n)}$}
    \begin{gathered}
    s^{(n)} \brackets{Q_k} 
    = Q_k, 
    \quad 
    1 \leq k \leq n - 1,
    \qquad
    s^{(n)} \brackets{Q_n} 
    = Q_n - \dfrac{b_n}{P_n}, 
    \\
    s^{(n)} \brackets{P_k} 
    = P_k, 
    \quad 
    1 \leq k \leq n,
    \qquad
    s^{(n)} \brackets{b_n} 
    = - b_n;
    \end{gathered}
    \\
    \tag{$r^{(n)}$}
    \begin{gathered}
    r^{(n)} \brackets{Q_k} 
    = b_{2k}^{(n)} 
    - \sum_{l = 1}^{n - k} 
    r^{(n)} \brackets{P_l} r^{(n)} \brackets{Q_{l + k}}, 
    \quad 
    1 \leq k \leq n - 1,
    \qquad
    r^{(n)} \brackets{Q_n} 
    = - Q_n, 
    \\
    r^{(n)} \brackets{P_k} 
    = P_k - \dfrac{1}{2^{2n - 1}} b_{2 \brackets{n - k} + 1}^{(n)}, 
    \quad 
    1 \leq k \leq n,
    \qquad
    r^{(n)} \brackets{b_n} 
    = - b_n - 2;
    \end{gathered}
\end{gather}
with fundamental relations
\begin{equation}
    \brackets{s^{(n)}}^2 = 1,
    \quad
    \brackets{r^{(n)}}^2 = 1,
\end{equation}
and preserves the canonicity of new variables. 

Its extension $\tilde{W}^{(n)}$ is given as 
\begin{equation}
    \tilde{W}^{(n)} = \angleBrackets{s_0^{(n)}, s_1^{(n)}; \pi^{(n)}},
    \quad 
    s_0^{(n)} = s^{(n)}, \, s_1^{(n)} = r^{(n)}s^{(n)}r^{(n)}, \, \pi^{(n)} = r^{(n)},
\end{equation}
with fundamental relations $\brackets{s_i^{(n)}}^2 = 1$, $\brackets{\pi^{(n)}}^2 = 1$, $\pi^{(n)} s_i^{(n)} = s_{i+1}^{(n)} \pi^{(n)}$, $i \in \quot{\mathbb{Z}}{2 \mathbb{Z}}$.
\end{thm}

One can see that $W^{(n)}$ acts on a parameter $\alpha_n = \frac{1}{2} \brackets{b_n + 1}$ by reflections with respect to zero and one-half:
\begin{figure}[H]
    \centering
    \scalebox{1.3}{\tikzset{every picture/.style={line width=0.75pt}} 

\begin{tikzpicture}[x=0.75pt,y=0.75pt,yscale=-1,xscale=1]

\draw    (51.6,50) -- (286.72,50) ;
\draw [shift={(288.72,50)}, rotate = 180] [color={rgb, 255:red, 0; green, 0; blue, 0 }  ][line width=0.75]    (10.93,-3.29) .. controls (6.95,-1.4) and (3.31,-0.3) .. (0,0) .. controls (3.31,0.3) and (6.95,1.4) .. (10.93,3.29)   ;
\draw  [color={rgb, 255:red, 251; green, 143; blue, 103 }  ,draw opacity=1 ][fill={rgb, 255:red, 251; green, 143; blue, 103 }  ,fill opacity=1 ] (174.76,50) .. controls (174.76,47.46) and (172.7,45.4) .. (170.16,45.4) .. controls (167.62,45.4) and (165.56,47.46) .. (165.56,50) .. controls (165.56,52.54) and (167.62,54.6) .. (170.16,54.6) .. controls (172.7,54.6) and (174.76,52.54) .. (174.76,50) -- cycle ;
\draw  [color={rgb, 255:red, 251; green, 143; blue, 103 }  ,draw opacity=1 ][fill={rgb, 255:red, 251; green, 143; blue, 103 }  ,fill opacity=1 ] (255.16,50) .. controls (255.16,47.46) and (253.1,45.4) .. (250.56,45.4) .. controls (248.02,45.4) and (245.96,47.46) .. (245.96,50) .. controls (245.96,52.54) and (248.02,54.6) .. (250.56,54.6) .. controls (253.1,54.6) and (255.16,52.54) .. (255.16,50) -- cycle ;
\draw  [color={rgb, 255:red, 104; green, 193; blue, 184 }  ,draw opacity=1 ][fill={rgb, 255:red, 104; green, 193; blue, 184 }  ,fill opacity=1 ] (214.76,49.6) .. controls (214.76,47.06) and (212.7,45) .. (210.16,45) .. controls (207.62,45) and (205.56,47.06) .. (205.56,49.6) .. controls (205.56,52.14) and (207.62,54.2) .. (210.16,54.2) .. controls (212.7,54.2) and (214.76,52.14) .. (214.76,49.6) -- cycle ;
\draw  [color={rgb, 255:red, 104; green, 193; blue, 184 }  ,draw opacity=1 ][fill={rgb, 255:red, 104; green, 193; blue, 184 }  ,fill opacity=1 ] (135.36,50.2) .. controls (135.36,47.66) and (133.3,45.6) .. (130.76,45.6) .. controls (128.22,45.6) and (126.16,47.66) .. (126.16,50.2) .. controls (126.16,52.74) and (128.22,54.8) .. (130.76,54.8) .. controls (133.3,54.8) and (135.36,52.74) .. (135.36,50.2) -- cycle ;
\draw  [color={rgb, 255:red, 251; green, 143; blue, 103 }  ,draw opacity=1 ][fill={rgb, 255:red, 251; green, 143; blue, 103 }  ,fill opacity=1 ] (94.56,49.8) .. controls (94.56,47.26) and (92.5,45.2) .. (89.96,45.2) .. controls (87.42,45.2) and (85.36,47.26) .. (85.36,49.8) .. controls (85.36,52.34) and (87.42,54.4) .. (89.96,54.4) .. controls (92.5,54.4) and (94.56,52.34) .. (94.56,49.8) -- cycle ;

\draw (294.9,42.9) node [anchor=north west][inner sep=0.75pt]  [font=\footnotesize]  {$\alpha_n$};
\draw (166.35,30) node [anchor=north west][inner sep=0.75pt]  [font=\scriptsize]  {$0$};
\draw (247.16,30) node [anchor=north west][inner sep=0.75pt]  [font=\scriptsize]  {$1$};
\draw (122.45,28) node [anchor=north west][inner sep=0.75pt]  [font=\scriptsize]  {$-\frac{1}{2}$};
\draw (82.95,30) node [anchor=north west][inner sep=0.75pt]  [font=\scriptsize]  {$-1$};
\draw (206.2,28) node [anchor=north west][inner sep=0.75pt]  [font=\scriptsize]  {$\frac{1}{2}$};
\draw (166.4,63) node [anchor=north west][inner sep=0.75pt]  [font=\footnotesize]  {$r$};
\draw (206.65,63) node [anchor=north west][inner sep=0.75pt]  [font=\footnotesize]  {$s$};
\draw (240.84,63) node [anchor=north west][inner sep=0.75pt]  [font=\footnotesize]  {$srs$};
\draw (121.07,63) node [anchor=north west][inner sep=0.75pt]  [font=\footnotesize]  {$rsr$};
\draw (75.73,63) node [anchor=north west][inner sep=0.75pt]  [font=\footnotesize]  {$rsrsr$};

\end{tikzpicture}}
    \caption*{$s \coloneqq s^{(n)}: \alpha_n \mapsto - \alpha_n + 1$, \quad $r \coloneqq r^{(n)}: \alpha_n \mapsto - \alpha_n$.}
\end{figure}

An extended affine Weyl group $\tilde{W}^{(n)}$ defines appropriate variables for a symmetric form for the $n$-th member of \eqref{nonstatPIIhier}.
\begin{thm}
A symmetric form for the $n$-th member of the non-stationary \text{\rm $\PIIn{n}$} hierarchy is defined by the following variables
\begin{align}
    f_{n - k + 1,0}^{(n)}
    &= P_k,
    &
    f_{n - k + 1,1}^{(n)}
    &= \pi^{(1)} \brackets{P_k},
    \quad 
    1 \leq k \leq n,
    &
    \pi^{(n)} \brackets{f_{k, i}^{(n)}}
    &= f_{k, i + 1},
    & 
    &
    \\
    g_{n - k,0}^{(n)}
    &= Q_k,
    &
    g_{n - k,1}^{(n)}
    &= \pi^{(1)} \brackets{Q_k},
    \quad 
    1 \leq k \leq n - 1,
    &
    \pi^{(n)} \brackets{g_{k, i}^{(n)}}
    &= g_{k, i + 1},
    & 
    i &\in \quot{\mathbb{Z}}{2 \mathbb{Z}},
    \\
    \beta_0^{(n)}
    &= b_n,
    &
    \beta_1^{(n)}
    &= \pi^{(1)} \brackets{b_n},
    \quad 
    &
    \pi^{(n)} \brackets{\beta_i^{(n)}}
    &=\beta_{i + 1}^{(n)},
    & 
    &
\end{align}
with the additional equation
\begin{equation}
    Q_n^{\prime}
    = 4^n \brackets{f_{n,0}^{(n)} - f_{n,1}^{(n)}},
\end{equation}
and the condition $\beta_0^{(n)} + \beta_1^{(n)} = - 2$.
\end{thm}

Since actions of the affine Weyl group commute with taking derivatives, variables defined above represent $2n$ equation of a Hamiltonian system as a symmetric system with respect to the $\pi^{(n)}$ action. To get from this system the $n$-th member of \eqref{nonstatPIIhier}, we should use the additional equation.

The \Backlund transformations are also helpful in the generating solutions process. By actions of \Backlund transformations compositions on the unique trivial solution $w \brackets{z, \overline{t}; 0} = 0$ (also called \textit{a seed solution}), we are able to obtain new different non-trivial solutions. For an integer parameter $\alpha_1 \coloneqq m \in \mathbb{Z}$ (in the case $n = 1$), there exists a correspondence between the second \Painleve equation $\PIIn{1}$ and meromorphic functions. Thus, in papers \cite{yablonskii1959rational} and \cite{vorobiev1965rational} it was shown that rational solutions arise from special polynomials called \textit{the Yablonskii-Vorobiev polynomials} that are defined by the difference-differential equation
\begin{align}
    u_{m + 1}^{\brackets{1}} u_{m - 1}^{\brackets{1}}
    &= z \brackets{u_m^{\brackets{1}}}^2
    - 4 \brackets{u_m^{\brackets{1}} \brackets{u_m^{\brackets{1}}}^{\prime\prime}
    - \brackets{\brackets{u_m^{\brackets{1}}}^{\prime}}^2},
    \\
    u_0^{(1)} &= 1,
    \quad 
    u_1^{(1)} = z,
\end{align}
which is also known as \textit{the Toda lattice equation}.

We have shown that there exists a one-to-one correspondence between the $n$-th member of \eqref{nonstatPIIhier} and meromorphic functions when $\alpha_n \coloneqq m \in \mathbb{Z}$. In this case solutions have the following form 
\begin{equation}
    w_m^{(n)} \brackets{z} = \dfrac{d}{dz} \ln \brackets{\dfrac{u_{m-1}^{(n)}}{u_m^{(n)}}},
    \quad 
    w_{-m}^{(n)} \brackets{z} = - w_m^{(n)} \brackets{z},
\end{equation}
where $u_m^{(n)}$ are  Yablonskii-Vorobiev-type polynomials associated with the $n$-th member of \eqref{nonstatPIIhier}.

Using \Backlund transformations and some properties of Yablonski-Vorobiev-type polynomials, we obtain a difference-differential equation for \textit{the Yablonskii-Vorobiev-type polynomials}.
\begin{thm}
Yablonskii-Vorobiev-type polynomials associated with the $n$-th member of \eqref{nonstatPIIhier} are defined by the following difference-differential equation
\begin{gather}
    u_{m + 1}^{\brackets{n}} u_{m - 1}^{\brackets{n}}
    = - 2 \brackets{u_m^{\brackets{n}}}^2 \sum_{l = 0}^{n} t_l 
    \mathcal{L}_l \LieBrackets{2 \dfrac{d^2}{d z^2} \brackets{\ln{u_m^{\brackets{n}}}}},
    \quad 
    t_0 = - z,
    \,
    t_n = 1,
\end{gather}
with $u_0^{\brackets{n}} = 1$ and $u_1^{\brackets{n}} = z$.
\end{thm}

\begin{rem}
   The difference-differential equation in the theorem above was also found in the work \cite{kudryashov2008generalized}, where polynomials $u_m^{(n)}$ were called \textit{the generalized Yablonskii-Vorobiev polynomials}. The author of the current paper have obtained this result independently of \cite{kudryashov2008generalized} and gives the most complete proof of this remarkable fact.
\end{rem}

\begin{rem}
   In Appendix \ref{YBplotsSection} we will present plots of Yablonskii-Vorobiev-type polynomials roots on the complex plane. There is an interesting observation that for big values of parameters $|t_i|$ shapes of plots of roots "reduce" to plots of some young series of Yablonskii-Vorobiev-type polynomials. 
\end{rem}

Rational solutions obtained by actions of \Backlund transformations on a unique trivial solution and by Yablonskii-Vorobiev-type polynomials have the following relation
\begin{equation}
    \dfrac{d}{dz} \ln \brackets{\dfrac{u_{m - 1}^{(n)}}{u_m^{(n)}}}
    = \frac{1}{4^n} T_m^{(n)} \brackets{Q_n},
    \quad
    T_m^{(n)}
    = \brackets{r^{(n)}s^{(n)}}^m.
\end{equation}

\begin{rem}
Since the second \Painleve hierarchy is a result of a self-similar reduction of the Korteveg-de Vriez hierarchy, rational solutions of \eqref{nonstatPIIhier} are also rational solutions of the Korteveg-de Vriez hierarchy
\begin{equation}
    u \brackets{x, t_{2 n + 1}}
    = 2 \dfrac{\partial^2}{\partial x^2}
    \PoissonBrackets{\ln \LieBrackets{u_m^{(n)} \brackets{\dfrac{x}{\brackets{(2 n + 1) t_{2 n + 1}}^{1/(2 n + 1)}}}}}.
\end{equation}
\end{rem}

Moreover, the Yablonskii-Vorobiev polynomials are connected with the polynomial $\tau$-function of the Kadomtsev–Petviashvili hierarchy \cite{jimbo1983solitons}. Due to this fact the Yablonskii-Vorobiev polynomials have determinant representations in the Jacobi-Trudi form and in the Hankel form \cite{kajiwara1996determinant}. Since hierarchy \eqref{nonstatPIIhier} can be also reduced from the Kadomtsev–Petviashvili hierarchy, we found out a generalization of polynomial $\tau$-functions and its Jacobi-Trudi determinant representation to the non-stationary case. 
\begin{thm}
Rational solutions of \eqref{nonstatPIIhier} are given as
\begin{equation}
    w_m^{(n)} \brackets{z} 
    = \dfrac{d}{dz} \PoissonBrackets{\ln \LieBrackets{\dfrac{\tau_{m - 1}^{(n)}}{\tau_m^{(n)}}}},
    \quad
    w_{- m}^{(n)} \brackets{z} 
    = - w_m^{(n)} \brackets{z},
    \quad 
    m \geq 1,
\end{equation}
where $\tau_m^{(n)} \brackets{z}$ are $m \times m$ determinants 
\begin{equation}
    \tau_m^{(n)} \brackets{z}
    = 
    \begin{vmatrix}
    h_{m}^{(n)} \brackets{z}
    &
    h_{m + 1}^{(n)} \brackets{z}
    &
    \dots 
    &
    h_{2m - 1}^{(n)} \brackets{z}
    \vspace{0.2cm}
    \\
    h_{m - 2}^{(n)} \brackets{z}
    &
    h_{m - 1}^{(n)} \brackets{z}
    &
    \dots 
    &
    h_{2m - 3}^{(n)} \brackets{z}
    \vspace{0.2cm}
    \\
    \vdots
    &
    \vdots
    &
    \ddots
    &
    \vdots
    \vspace{0.2cm}
    \\
    h_{- m + 2}^{(n)} \brackets{z}
    &
    h_{- m + 3}^{(n)} \brackets{z}
    &
    \dots 
    &
    h_{1}^{(n)} \brackets{z}
    \end{vmatrix},
\end{equation}
which entries are the complete homogeneous polynomials $h_k^{(n)} \brackets{z}$, $h_k^{(n)} \brackets{z} = 0$ for $k < 0$ defined by a generating function in a formal parameter $\lambda$
\begin{equation}
    \sum_{k = 0}^{\infty} 
    h_k^{(n)} \brackets{z} \lambda^k
    = \exp{\brackets{- \sum_{l = 0}^n \dfrac{4^l}{2 l + 1} t_{l} \lambda^{2 l + 1}}}, 
    \quad 
    t_0 = - z,
    \quad 
    t_n = 1.
\end{equation}
\end{thm}

As a consequence, we have a relation between the Yablonskii-Vorobiev-type polynomials and the polynomial $\tau^{(n)}$-functions of \eqref{nonstatPIIhier}, which does not depend on the member number $n$.
\begin{corl}
$\tau_m^{(n)} \brackets{z}
        = c_m u_m^{(n)} \brackets{z},
        \quad 
        c_m = \prod_{k = 1}^m \brackets{2 k + 1}^{k - m}$.
\end{corl}

This paper is organized as follows. In Section \ref{BTSection}, we consider auto-\Backlund transformations of the $n$-th member of the non-stationary second \Painleve hierarchy \eqref{nonstatPIIhier}, which are a generalization of results in the paper \cite{clarkson1999backlund}. In Section \ref{AWsection}, we give explicit formulas of generators in canonical coordinates defined in \cite{mazzocco2007hamiltonian} of an affine Weyl group $W^{(n)}$ associated with the $n$-th member of \eqref{nonstatPIIhier} and its extension $\tilde{W}^{(n)}$. In Section \ref{YVsection}, we recall a definition of the classical Yablonskii-Vorobiev polynomials $u_m^{(1)} \brackets{z}$ obtained in \cite{yablonskii1959rational, vorobiev1965rational} and construct their generalization $u_m^{(n)} \brackets{z}$ to the non-stationary case for the $n$-th member of \eqref{nonstatPIIhier}. In Section \ref{TFsection}, we study polynomial $\tau^{(n)}$-functions of the $n$-th member \eqref{nonstatPIIhier} and its Jacobi-Trudi determinant representation via complete homogeneous polynomials $h_k^{(n)} \brackets{z}$. As a consequence, Yablonskii-Vorobiev-type polynomials have the same determinant representation. In each section we follow all details of our results in explicit examples for the first and second members of the non-stationary $\PIIn{n}$ hierarchy \eqref{nonstatPIIhier}.

\textbf{Acknowledgements.}
The author is so grateful to her scientific advisors Vladimir Poberezhnyi and Vladimir Rubtsov, who introduced her to this amazing area of mathematical physics, constantly supported the author during her scientific work, and developed the author's interest in the \Painleve equations theory. The author would like to express her special gratitude to Vladimir Rubtsov for setting the scientific problem and his continuous attention to the author's work. The author is also grateful to Marta Mazzocco, who drew the author's attention to a different way of constructing generators of an affine Weyl group, and to Ilia Gaiur for fruitful discussions.

The current paper is a part of the author's PhD program studies at HSE University (HSE) and has been carrying out at the Faculty of Mathematics. The author would like to thank this faculty for giving her such opportunity. This paper was partially supported under the Grant RFBR 18-01-00461 A of the Russian Foundation for Basic Research and the International Laboratory of Cluster Geometry HSE, RF Government grant № 075-15-2021-608.

The author thanks the referee for the careful reading of the paper and helpful remarks, which, in particular, improved the proof of Theorem \ref{BTthm}. 

\section{Auto-\Backlund transformations} \label{BTSection}

We will call \textit{a \Backlund transformation} transforming one PDE solution to another its solution. \textit{An auto-\Backlund transformation} leaves a PDE class invariant. As an example of a \Backlund transformation, we can consider the Miura transformation $u \brackets{x, T_3} = v_x \brackets{x, T_3} - v^2 \brackets{x, T_3}$ which transforms solutions of the Korteveg-de Vriez equation into solutions of the modified Korteveg-de Vriez equation
\begin{equation}
    u_{T_3} - 6 u u_x + u_{xxx} = 0
    \quad 
    \Rightarrow
    \quad 
    v_{T_3} - 6 v^2 v_x + v_{xxx} = 0.
\end{equation}

The discrete symmetry $w \brackets{z; - \alpha_1 } = - w \brackets{z; \alpha_1 }$ of the second \Painleve equation \eqref{PII1} is an example of an auto-\Backlund transformation. The less trivial example of an auto-\Backlund transformation of \eqref{PII1} is the well-known map
\begin{equation} \label{BT1}
    \tilde{w} \brackets{z; \tilde{\alpha}_1}
    = w \brackets{z; \alpha_1}
    - \dfrac{2 \alpha_1 - \varepsilon}{2 \varepsilon w^{\prime} \brackets{z; \alpha_1} - 2 w^2 \brackets{z; \alpha_1} - z},
    \quad 
    \tilde{\alpha}_1 = 1 - \alpha_1,
    \quad 
    \varepsilon = \pm 1,
\end{equation}
where $\tilde{w} \brackets{z; \tilde{\alpha}_1}$ is a solution of $\PIIn{1} \LieBrackets{\tilde{w} \brackets{z; \tilde{\alpha}_1}}$. Note that the denominator can be rewritten in terms of the Lenard operators as $2 \varepsilon w^{\prime} - 2 w^2 - z = 2 \sum_{l = 0}^1 t_l \mathcal{L} \LieBrackets{\varepsilon w^{\prime} - w^2}$ with $t_0 = - z$ and $t_1 = 1$. In the case $\alpha_1 = \frac{1}{2}$, the second \Painleve equation \eqref{PII1} has a special integral defined as
\begin{equation}
    I_{1/2}^{(1)}:
    \quad 
    2 w^{\prime} - 2 w^2 - z = 0,
\end{equation}
which satisfies the relation
\begin{equation}
    \brackets{
    \dfrac{d}{dz} + 2 w
    } I_{1/2}^{(1)}
    = 0.
\end{equation}
Using \Backlund transformations \eqref{BT1}, one can generate special integrals for all half-integer parameters $\alpha_1~\in~\mathbb{Z}~+~\frac{1}{2}$.

In this Section we are interested in auto-\Backlund transformations of \eqref{nonstatPIIhier}.

\begin{thm} \label{BTthm}
    Auto-\Backlund transformations of the $n$-th member \eqref{nonstatPIIhier} are defined for $\alpha_n \neq \frac{1}{2}$ and are~given~as
    \begin{align} \label{nonstatBT}
        \Tilde{w} \brackets{z; \tilde{\alpha}_n}
        &= w \brackets{z; \alpha_n} 
        - \dfrac{2 \alpha_n - \varepsilon}{2 \dsum_{l = 0}^{n} t_l \mathcal{L}_l \LieBrackets{\varepsilon w^{\prime} - w^2}},
        \\
        &
        \qquad
        t_0 = -z,
        \,\, 
        t_n = 1,
        \quad 
        \tilde{\alpha}_n = 1 - \alpha_n,
        \quad
        \varepsilon = \pm 1.
    \end{align}
\end{thm}

\begin{corl}
Special integrals $I_{k/2}^{(n)}$, $k \in \mathbb{Z}$, of \eqref{nonstatPIIhier} can be obtained from the special integral $I_{1/2}^{(n)}$,
\begin{equation}
    I_{1/2}^{(n)}:
    \quad 
    \sum_{l = 0}^n t_l \mathcal{L}_l \LieBrackets{w^{\prime} - w^2}
    = 0,
    \quad 
    t_0 = -z,
    \,\, 
    t_n = 1,
\end{equation}
by actions of \Backlund transformations \eqref{nonstatBT} on it.
\end{corl}

\begin{proof}[Proof of Theorem \ref{BTthm}]

Let us fix the number $n$ of the member \eqref{nonstatPIIhier}. Let $w(z; \alpha_n) = w$ be a solution of the $n$-th member of \eqref{nonstatPIIhier}. We define a function $\psi$ as
\begin{align}
    \label{eq:psidef}
    w 
    &= \frac12 (\ln \psi')'.
\end{align}
Then $w' - w^2 = \tfrac12 S (\psi)$, where $S(\psi)$ is the Schwartz derivative which is invariant under the \Mobius transformation $\psi \mapsto - \psi^{-1}$. Due to this fact, we take a function $\tilde{w}$ related to the solution $w$ by the \Mobius transformation:
\begin{align} \label{eq:tdw}
    \tilde{w}
    &= \frac12 \brackets{\ln (- \psi^{-1})'}'
    = \frac12 \psi'' \, (\psi')^{-1}
    - \psi' \, \psi^{-1}
    = w 
    - \psi' \, \psi^{-1} 
    .
\end{align}

Let us show that $\tilde w$ is a solution of \eqref{nonstatPIIhier} with another parameter $\tilde \alpha_n \neq \alpha_n$:
\begin{align}
    &&
    \brackets{\dfrac{d}{dz} + 2 \Tilde{w}} 
    \sum_{l = 0}^{n} t_l \mathcal{L}_l \LieBrackets{\Tilde{w}^{\prime} - \Tilde{w}^2}
    &= \tilde{\alpha}_n - \frac12,
    &
    t_0
    &= - z.
    &&
\end{align}
Substituting \eqref{eq:tdw} into the equation above and using the fact that $w$ is a solution of \eqref{nonstatPIIhier} with $\alpha_n$ and the invariance of the Schwartz derivative under the \Mobius transformation, we obtain:
\begin{gather}
    \brackets{\dfrac{d}{dz} 
    + 2 {w} - 2 \psi' \, \psi^{-1}
    } 
    \sum_{l = 0}^{n} t_l \mathcal{L}_l \LieBrackets{{w}^{\prime} - {w}^2}
    = \tilde \alpha_n - \frac12,
    \\[1mm]
    \alpha_n - \tilde{\alpha}_n
    = 2 \psi' \, \psi^{-1} 
    \sum_{l = 0}^{n} t_l \mathcal{L}_l \LieBrackets{{w}^{\prime} - {w}^2}.
\end{gather}
The r.h.s. is equal to $2 \alpha_n - 1$, since $w$ is a solution of \eqref{nonstatPIIhier} and the function $\psi$ is defined up to a shift. Indeed, with the help of \eqref{eq:psidef}, equation \eqref{nonstatPIIhier} can be rewritten as
\begin{gather}
    \dfrac{d}{dz} 
    \sum_{l = 0}^{n} t_l \mathcal{L}_l \LieBrackets{{w}^{\prime} - {w}^2}
    + \psi'' \, (\psi')^{-1} 
    \sum_{l = 0}^{n} t_l \mathcal{L}_l \LieBrackets{{w}^{\prime} - {w}^2}
    = {\alpha}_n - \frac12,
\end{gather}
or, equivalently,
\begin{gather}
    \dfrac{d}{dz} \brackets{
    \psi' \sum_{l = 0}^{n} t_l \mathcal{L}_l \LieBrackets{{w}^{\prime} - {w}^2}
    }
    = \dfrac{d}{dz} \brackets{
    \brackets{{\alpha}_n - \frac12}
    \psi
    }.
\end{gather}
Integrating the latter w.r.t. $z$ and shifting the function $\psi$, we get
\begin{gather}
    \psi'  
    \sum_{l = 0}^{n} t_l \mathcal{L}_l \LieBrackets{{w}^{\prime} - {w}^2}
    = \brackets{\alpha_n - \frac12} \, \psi.
\end{gather}
Therefore, $\tilde w$ is a solution of \eqref{nonstatPIIhier} for $\tilde \alpha_n = 1 - \alpha_n$ and two solutions $w (z; \alpha_n)$ and $\tilde w (z; 1 - \alpha_n)$ of the $n$-th member of \eqref{nonstatPIIhier} are related to each other by the following transformation
\begin{gather}
    \tilde w
    = w 
    - \dfrac{2 {\alpha}_n - 1}{2\sum_{l = 0}^{n} t_l 
    \mathcal{L}_l \LieBrackets{{w}^{\prime} - {w}^2}}.
\end{gather}
Since the $n$-th member of \eqref{nonstatPIIhier} has the discrete symmetry $w \brackets{z; - \alpha_n} \mapsto - w \brackets{z; \alpha_n}$, we arrive at the statement of our theorem. 
\end{proof}
    
\begin{rem}
    The proof uses the same ideas as in the derivation of \Backlund transformations in the stationary case $\bar{t} = 0$, which were obtained in \cite{clarkson1999backlund}.
\end{rem}

\section{Affine Weyl groups and their extensions} \label{AWsection}

The classical result of the second \Painleve equation analysis through its symmetries is that these symmetries form a special group structure associated with an affine Weyl group of type $\typeA{1}$. Let us consider two generators of \Backlund transformations \eqref{BT1}:
\begin{align} \label{BT1gen}
    \begin{aligned}
    s^{(1)}:
    &
    \qquad 
     \Tilde{w} \brackets{z; \tilde{\alpha}_1}
        = w \brackets{z; \alpha_1} 
        - \dfrac{2 \alpha_1 - 1}{2 w^{\prime} - 2 w^2 - z},
    \\
    r^{(1)}:
    &
    \qquad 
     \Tilde{w} \brackets{z; -{\alpha}_1}
        = - w\brackets{z; {\alpha}_1}.
    \end{aligned}
\end{align}

Using the Hamiltonian structure obtained in \cite{mazzocco2007hamiltonian} for the first member of \eqref{nonstatPIIhier}:
\begin{gather} \label{HS1}
    \begin{gathered}
    Q_1 = 4 w,
    \quad 
    P_1 = \dfrac{1}{2} \brackets{w^{\prime} - w^2 - \dfrac{z}{2}},
    \\
    \PoissonBrackets{Q_1, P_1} = 1,
    \quad
    \PoissonBrackets{Q_1, Q_1}
    = \PoissonBrackets{P_1, P_1}
    = 0;
    \\
    \HIIn{1} \brackets{\alpha_1}
    = 4 P_1 \brackets{P_1 + \dfrac{1}{16} Q_1^2 + \dfrac{z}{2}}
    - \dfrac{1}{4} Q_1 \brackets{2 \alpha_1 - 1};
    \end{gathered}
\end{gather}
we can rewrite actions of \Backlund generators \eqref{BT1gen} on canonical coordinates in the following way
\begin{align} \label{BT1genQP}
    \begin{aligned}
    s^{(1)} \brackets{Q_1}
    &= Q_1 - \dfrac{b_1}{P_1},
    &&&&&
    s^{(1)} \brackets{P_1}
    &= P_1,
    &&&&&
    s^{(1)} \brackets{b_1}
    &= - b_1;
    \\
    r^{(1)} \brackets{Q_1}
    &= - Q_1,
    &&&&&
    r^{(1)} \brackets{P_1}
    &= - \brackets{P_1 + \dfrac{1}{16} Q_1^2 + \dfrac{z}{2}},
    &&&&&
    r^{(1)} \brackets{b_1}
    &= - b_1 - 2.
    \end{aligned}
\end{align}

Note that
\begin{itemize}
    \item new coordinates are also canonical,
    \item $\brackets{s^{(1)}}^2 = 1$ and $\brackets{r^{(1)}}^2 = 1$,
    \item compositions $T_m^{(1)} = \brackets{r^{(1)} s^{(1)}}^m$ and $S_k^{(1)}  = \brackets{r^{(1)} s^{(1)}}^k s^{(1)}$ have a group structure and generate all \Backlund transformations \eqref{BT1}.
\end{itemize}

Therefore, \textit{compositions of \Backlund transformations \eqref{BT1} form the affine Weyl group $W^{(1)} = \angleBrackets{s^{(1)}, r^{(1)}}$ of type $\typeA{1}$ with fundamental relations $\brackets{s^{(1)}}^2 = 1$ and $\brackets{r^{(1)}}^2 = 1$}. The corresponding Cartan matrix and the Dynkin diagram are
    \begin{table}[H]
        \centering
        \begin{tabular}{cc}
        The Cartan matrix & The Dynkin diagram
        \\
        $A=
        \begin{pmatrix}
            2 & - 2\\
            - 2 & 2
        \end{pmatrix}
        \quad
        $
        & 
        \begin{tabular}{c}
        \\
        \tikzset{every picture/.style={line width=0.75pt}} 

\begin{tikzpicture}[x=0.75pt,y=0.75pt,yscale=-1,xscale=1]

\draw    (200.91,48.29) -- (240.16,48.29) ;
\draw    (200.47,51.91) -- (241.48,51.91) ;
\draw  [color={rgb, 255:red, 104; green, 193; blue, 184 }  ,draw opacity=1 ][fill={rgb, 255:red, 104; green, 193; blue, 184 }  ,fill opacity=1 ] (205.36,50) .. controls (205.36,47.46) and (203.3,45.4) .. (200.76,45.4) .. controls (198.22,45.4) and (196.16,47.46) .. (196.16,50) .. controls (196.16,52.54) and (198.22,54.6) .. (200.76,54.6) .. controls (203.3,54.6) and (205.36,52.54) .. (205.36,50) -- cycle ;
\draw   (210.05,54.33) -- (201.62,49.8) -- (210.3,45.77) ;
\draw  [color={rgb, 255:red, 251; green, 143; blue, 103 }  ,draw opacity=1 ][fill={rgb, 255:red, 251; green, 143; blue, 103 }  ,fill opacity=1 ] (244.76,50.27) .. controls (244.76,47.73) and (242.7,45.67) .. (240.16,45.67) .. controls (237.62,45.67) and (235.56,47.73) .. (235.56,50.27) .. controls (235.56,52.81) and (237.62,54.87) .. (240.16,54.87) .. controls (242.7,54.87) and (244.76,52.81) .. (244.76,50.27) -- cycle ;
\draw   (231.02,45.84) -- (239.58,50.12) -- (231.02,54.4) ;


\end{tikzpicture}
        \end{tabular}
        \end{tabular}
    \end{table}

An extended affine Weyl group $\tilde{W}^{(1)} = \angleBrackets{s_0^{(1)}, s_1^{(1)}; \pi^{(1)}}$ has generators $s_0^{(1)} = s^{(1)}$, $s_1^{(1)} =  r^{(1)} s^{(1)} r^{(1)}$, $\pi^{(1)} = r^{(1)}$ with fundamental relations
\begin{equation}
    \brackets{s_i^{(1)}}^2 
    = 1,
    \quad 
    \brackets{\pi^{(1)}}^2 
    = 1,
    \quad 
    \pi^{(1)} s_i^{(1)}
    = s_{i + 1}^{(1)} \pi^{(1)},
    \quad 
    i \in \quot{\mathbb{Z}}{2 \mathbb{Z}}.
\end{equation}

\begin{rem}
    Actions of the affine Weyl group $W^{(1)}$ on canonical coordinates and a parameter $b_1$ define the discrete dynamics  d-P$\brackets{{A_1^{\brackets{1}}}/{E_7^{\brackets{1}}}}$ \cite{kajiwara2017geometric}.
\end{rem}

In this section we find out an exact form of generators of an affine Weyl group in terms of canonical coordinates for the $n$-th member of \eqref{nonstatPIIhier}. We will work with two different generators $s^{(n)}$ and $r^{(n)}$ of \Backlund transformations \eqref{nonstatBT} defined by the following data
\begin{align} \label{BTgen}
    \begin{aligned}
    s^{(n)}:
    &
    \qquad 
     \Tilde{w} \brackets{z, \overline{t}; \tilde{\alpha}_n}
        = w \brackets{z, \overline{t}; \alpha_n} 
        - \dfrac{2 \alpha_n - 1}{2 \dsum_{l = 1}^{n} t_l \mathcal{L}_l \LieBrackets{w^{\prime} - w^2} - z},
    \\
    r^{(n)}:
    &
    \qquad 
     \Tilde{w} \brackets{z, \overline{t}; -{\alpha}_n}
        = - w\brackets{z, \overline{t}; {\alpha}_n}.
    \end{aligned}
\end{align}

\begin{rem}
    The composition $r^{(n)} s^{(n)}$ gives the \Backlund transformation \eqref{nonstatBT} with $\varepsilon = - 1$.
\end{rem}

To express \eqref{BTgen} in terms of canonical coordinates $Q_i$, $P_j$, $1 \leq i, j \leq n$, with the standard Poisson bracket
\begin{equation}
    \PoissonBrackets{Q_i, P_j} = \delta_{ij},
    \quad 
    \PoissonBrackets{Q_i, Q_j}
    = \PoissonBrackets{P_i, P_j} = 0,
\end{equation}
we will express canonical coordinates via $w \brackets{z; \alpha_n}$ and its derivatives, then rewrite actions of \eqref{BTgen} on canonical coordinates in terms of matrix entries  $b_{2i+1}^{(n)}$ and $b_{2i}^{(n)}$ of a matrix $\mathcal{A}^{(n)}$ of isomonodromic deformation problem (15) in \cite{mazzocco2007hamiltonian}. Finally, by Theorem 6.1 in \cite{mazzocco2007hamiltonian}, these matrix entries are polynomials of canonical coordinates. \footnote{Many thanks to Marta Mazzocco, who drew the author's attention to this fact!}
\subsection{Canonical coordinates as polynomials of the \texorpdfstring{$\PIIn{n}$}{PIIn} solution}

Under the definition of canonical coordinates in Theorem 5.1 in \cite{mazzocco2007hamiltonian}, we are able to express canonical coordinates through a solution of \eqref{nonstatPIIhier} and its derivatives.

\begin{thm} \label{canCoordAsPols}
Canonical coordinates $Q_i$, $P_i$, $1 \leq i \leq n$, for the $n$-th member of the non-stationary \text{\rm $\PIIn{n}$} hierarchy can be represented as polynomials of the \text{\rm $\PIIn{n}$} solution $w \brackets{z, \overline{t}; \alpha_n}$ and its derivatives in the following form
\begin{gather} \label{canCoordAsPolsEq}
    \begin{gathered}
    \begin{aligned}
    Q_n 
    &= 4^n w,
    &&&&&
    P_n
    &= \dfrac{1}{2^{2n - 1}} \sum_{l = 0}^n t_{l} \mathcal{L}_{l} \LieBrackets{w^{\prime} - w^2};
    \end{aligned}
    \\
    Q_{k}
    = 4^{k} \brackets{\dfrac{d}{dz} + 2 w} 
    \dsum_{l = 1}^n t_{l} \mathcal{L}_{l - k} \LieBrackets{w^{\prime} - w^2}
    - \sum_{l = 1}^{n - k} P_l Q_{l + k
    },
    \\
    P_k 
    = \dfrac{1}{2^{2k - 1}} 
    \sum_{l = 0}^k t_{n - l} 
    \mathcal{L}_{k - l} \LieBrackets{w^{\prime} - w^2},
    \quad 
    1 \leq k \leq n - 1,
    \end{gathered}
\end{gather}
with $t_0 = -z$, $t_n = 1$.
\end{thm}
\begin{proof}
\phantom{}

\begin{itemize}
    \item for $P_n$:
    \begin{align}
        P_n
        &= \dfrac{a_{1}^{(n)} + b_{1}^{(n)}}{a_{2n+1}^{(n)}}
        = \dfrac{1}{4^n}
        \dsum_{l = 1}^n t_l 
        \brackets{A_{1}^{(l)} + B_{1}^{(l)} - z}
        \\
        &= \dfrac{1}{2^{2n - 1}}
        \sum_{l = 1}^n t_{l} 
        \mathcal{L}_{l} \LieBrackets{w^{\prime} - w^2}
        - \dfrac{z}{4^n}
        = \dfrac{1}{2^{2n - 1}}
        \sum_{l = 0}^n t_{l} 
        \mathcal{L}_{l} \LieBrackets{w^{\prime} - w^2}
        .
    \end{align}
    
    \item for $P_k$, $1 \leq k < n$,  we have
    \begin{align}
        P_k 
        &= \dfrac{a_{2(n-k)+1}^{(n)} + b_{2(n-k)+1}^{(n)}}{a_{2n+1}^{(n)}}
        = \dfrac{1}{4^n}
        \dsum_{l = 1}^n t_l 
        \brackets{A_{2 (n - k) + 1}^{(l)} + B_{2 (n - k) + 1}^{(l)}}
        \\
        &= \dfrac{1}{2} \dfrac{4^{n - k + 1}}{4^n} 
        \sum_{l = 1}^n t_l \mathcal{L}_{l - n + k}
        \LieBrackets{w^{\prime} - w^2}
        = \dfrac{1}{2^{2k - 1}} 
        \brackets{
        \mathcal{L}_{k} \LieBrackets{w^{\prime} - w^2}
        + 
        \sum_{l = 1}^k t_{n - l} 
        \mathcal{L}_{k - l} \LieBrackets{w^{\prime} - w^2}}
        .
    \end{align}
    
    \item for $Q_n$:
    \begin{align}
        Q_n 
        &= - \sum_{j = 1}^n b_{2j}^{(n)} 
        \LieBrackets{\brackets{\sum_{i = 0}^n P_i \gamma^{2i}}^{-1}}_{2j - 2n}
        = - b_{2n}^{(n)}.
    \end{align}
    
    \item To prove the last formula for $Q_k$, $1 \leq k \leq n - 1$, we rewrite it in the following way
    \begin{align}
        Q_k
        &= - b_{2k}^{(n)}
        - \sum_{l = 1}^{n - k} P_l Q_{l + k}
        \\
        \sum_{l = 0}^{n - k} P_l Q_{l + k} 
        &= - b_{2k}^{(n)}
        ,
    \end{align}
    where we set $P_0 = 1$.
    
    The last relation can be represented in the matrix form
    \begin{equation}
        Q 
        = P^{-1} \cdot B,
    \end{equation}
    where 
    \begin{equation}
        P 
        = 
        \begin{pmatrix}
        P_0 & P_1 & P_2 & \dots & P_{n - 1}
        \\
        0 & P_0 & P_1 & \dots & P_{n - 2}
        \\
        0 & 0 & P_0 & \dots & P_{n - 3}
        \\
        \vdots & \vdots & \vdots & \ddots & \vdots
        \\
        0 & 0 & 0 & \dots & P_0
        \end{pmatrix}
        = \brackets{P_{j - k}}_{j > k; j, k = 0}^n
        ,
        \quad 
        Q 
        = 
        \begin{pmatrix}
        Q_1
        \\
        Q_2
        \\
        Q_3
        \\
        \vdots
        \\
        Q_n
        \end{pmatrix},
        \quad 
        B
        = 
        \begin{pmatrix}
        - b_2^{(n)}
        \\
        - b_4^{(n)}
        \\
        - b_6^{(n)}
        \\
        \vdots
        \\
        - b_{2n}^{(n)}
        \end{pmatrix}
        .
    \end{equation}
    
    \begin{lem} \label{upperTrMat}
        Let $P$ be the $n \times n$ matrix defined above. Then entries of its inverse matrix coincide with corresponding coefficients of the sequence $\brackets{\dsum_{i = 0}^{n - 1} P_i \gamma^{2i}}^{-1}$ in the given way 
        \begin{equation}
            \brackets{P^{-1}}_{j,k}
            = \LieBrackets{\brackets{\sum_{i = 0}^{n - 1} P_i \gamma^{2i}}^{-1}}_{2j - 2k}.
        \end{equation}
    \end{lem}
    \begin{proof}[Proof of Lemma \ref{upperTrMat}]
    On the one hand, the given sequence can be represented as
    \begin{align}
        \brackets{\dsum_{i = 0}^{n - 1} P_i {\gamma}^{2i}}^{-1}
        &= \brackets{1 + \dsum_{i = 1}^{n - 1} P_i {\gamma}^{2 i}}^{-1}
        \\
        \label{seqExpr}
        &= \sum_{m \geq 0} 
        \brackets{- 1}^m 
        \brackets{\dsum_{i = 1}^{n - 1} P_i {\gamma}^{2i}}^m.
    \end{align}
    
    On the other hand, the inverse matrix is
    \begin{align} 
        P^{-1}
        &= \brackets{\mathbb{E} + N}^{-1}
        = \sum_{m = 0}^{n - 1} \brackets{-1}^m N^m
        \\
        \label{P^(-1)Expr}
        &= \sum_{m = 0}^{n - 1} \brackets{-1}^m
        \brackets{\sum_{i = 1}^{n - 1} P_i J^i}^m
        .
    \end{align}
    where $N$ is a nilpotent matrix with index $n$ and $J$ is a Jordan block
    \begin{equation}
        N 
        = 
        \begin{pmatrix}
        0 & P_1 & P_2 & \dots & P_{n - 1}
        \\
        0 & 0 & P_1 & \ddots & P_{n - 2}
        \\
        \vdots & \vdots & \ddots & \ddots & \vdots
        \\
        0 & 0 & 0 & \dots & P_1
        \\
        0 & 0 & 0 & \dots & 0
        \end{pmatrix}
        ,
        \quad
        J 
        = 
        \begin{pmatrix}
        0 & 1 & 0 & \dots & 0
        \\
        0 & 0 & 1 & \ddots & 0
        \\
        \vdots & \vdots & \ddots & \ddots & \vdots
        \\
        0 & 0 & 0 & \dots & 1
        \\
        0 & 0 & 0 & \dots & 0
        \end{pmatrix}
        .
    \end{equation}
    
    Comparing \eqref{seqExpr} with \eqref{P^(-1)Expr}, we have the lemma statement.
    \end{proof}
\end{itemize}
\end{proof}

\begin{exmp}
\phantom{}

For the case $n = 1$, we have formulas \eqref{HS1}:
\begin{equation}
    Q_1 
    = 4 w, 
    \quad 
    P_1 
    = \dfrac{1}{2} \brackets{w^{\prime} - w^2 - \dfrac{z}{2}}.
\end{equation}

For the case $n = 2$, canonical coordinates in terms of $w \brackets{z, \overline{t}; \alpha_2}$ and its derivatives are
\begin{gather}
    Q_2 = 16 w, 
    \quad
    P_2 = \dfrac{1}{16} \brackets{ - z + 6 w^4 - 12 w^2 w^{\prime} + 2 {w^{\prime}}^2 - 4 w w^{\prime\prime} + 2 w^{\prime\prime\prime}
    + 2 t_1 \brackets{w^{\prime} - w^2}};
    \\
    Q_1 = - 8 w w^{\prime} + 4 w^{\prime\prime}, 
    \quad 
    P_1 = \dfrac{1}{2} \brackets{w^{\prime} - w^2 + \dfrac{t_1}{2}}.
\end{gather}
\end{exmp}
\subsection{Generators of affine Weyl groups}

In the current subsection we will express actions of generators \eqref{BTgen} on expressions of canonical coordinates defined in Theorem \ref{canCoordAsPols} in terms of matrix entries $b_{2 i + 1}^{(n)}$ and $b_{2 i}^{(n)}$. Since we know how these entries depend on  canonical coordinates, we find out generators of an affine Weyl group in terms of canonical coordinates $Q_i$, $P_i$, $i = 1, \dots, n$.

\begin{thm} \label{WeilConnection}
An affine Weyl group $W^{(n)}$ associated with the $n$-th member of \eqref{nonstatPIIhier} has a type $A_1^{(1)}$, is~given~as 
\begin{equation}
    W^{(n)} = \angleBrackets{s^{(n)}, r^{(n)}}, 
\end{equation}
where generators $s^{(n)}$ and $r^{(n)}$ act on canonical coordinates and a parameter $b_n = 2 \alpha_n - 1$ by following rules
\begin{gather} 
    \label{sGenW}
    \tag{$s^{(n)}$}
    \begin{gathered}
    s^{(n)} \brackets{Q_k} 
    = Q_k, 
    \quad 
    1 \leq k \leq n - 1,
    \qquad
    s^{(n)} \brackets{Q_n} 
    = Q_n - \dfrac{b_n}{P_n}, 
    \\
    s^{(n)} \brackets{P_k} 
    = P_k, 
    \quad 
    1 \leq k \leq n,
    \qquad
    s^{(n)} \brackets{b_n} 
    = - b_n;
    \end{gathered}
    \\
    \label{rGenW}
    \tag{$r^{(n)}$}
    \begin{gathered}
    r^{(n)} \brackets{Q_k} 
    = b_{2k}^{(n)} 
    - \sum_{l = 1}^{n - k} 
    r^{(n)} \brackets{P_l} r^{(n)} \brackets{Q_{l + k}}, 
    \quad 
    1 \leq k \leq n - 1,
    \qquad
    r^{(n)} \brackets{Q_n} 
    = - Q_n, 
    \\
    r^{(n)} \brackets{P_k} 
    = P_k - \dfrac{1}{2^{2n - 1}} b_{2 \brackets{n - k} + 1}^{(n)}, 
    \quad 
    1 \leq k \leq n,
    \qquad
    r^{(n)} \brackets{b_n} 
    = - b_n - 2;
    \end{gathered}
\end{gather}
with fundamental relations $\brackets{s^{(n)}}^2 = 1$, $\brackets{r^{(n)}}^2 = 1$, and preserves the canonicity of variables.
\end{thm}
\begin{proof}
\phantom{}

\textbf{1. Expressions of canonical coordinates after actions of generators on them.}

The proof is constructed on generators \eqref{BTgen}. Since $Q_n = 4^n w$, generators $r^{(n)}$ and $s^{(n)}$ define two base transformations
\begin{equation}
    s^{(n)} \brackets{Q_n} = Q_n - \dfrac{b_n}{P_n},
    \qquad
    r^{(n)} \brackets{Q_n} = - Q_n.
\end{equation}

Using this formulas, Theorem 5.1 in \cite{mazzocco2007hamiltonian}, and Theorem \ref{canCoordAsPols} in the current paper, we find out actions of generators on remaining canonical coordinates.

\vspace{0.5cm}
\textbf{$s^{(n)}$-transformation.} 

To prove formulas for $s^{(n)}$, the following lemma is helpful.
\begin{lem} \label{sInvariance}
Under the $s^{(n)}$-transformation the expression $Q_n^{\prime} - Q_n^2$ is invariant.
\end{lem}
\begin{proof}[Proof of Lemma \ref{sInvariance}]

This follows from the fact that $w^{\prime} - w^2$ is invariant under \Backlund transformation~\eqref{nonstatBT}.
\end{proof}

Hence, $s^{(n)}$ acts trivially on all canonical coordinates except of $Q_n$.

To find how $b_n$ transforms, we use the Hamiltonian defined by Theorem 7.2 in the work \cite{mazzocco2007hamiltonian}. The Hamiltonian system contains only one equation which depends on a parameter $b_n$:
\begin{equation}
    P_n^{\prime} 
    = - \dfrac{1}{2^{2n - 1}} Q_n P_n 
    + \dfrac{1}{4^n} b_n,
\end{equation}
one can obtain this just rewriting the $n$-th member \eqref{nonstatPIIhier} via canonical coordinates.

Then
\begin{align}
    \brackets{s^{(n)} \brackets{P_n}}^{\prime}
    = P_n^{\prime}
    &= - \dfrac{1}{2^{2n - 1}} Q_n P_n 
    + \dfrac{1}{4^n} b_n
    \\
    &= - \dfrac{1}{2^{2n - 1}} \brackets{Q_n - \dfrac{b_n}{P_n}} P_n
    - \dfrac{1}{4^n} b_n
    \\
    &= - \dfrac{1}{2^{2n - 1}} s^{(n)} \brackets{Q_n} s^{(n)} \brackets{P_n}
    + \dfrac{1}{4^n} s^{(n)} \brackets{b_n},
\end{align}
where $s^{(n)} \brackets{b_n} = - b_n$.

\begin{rem}
Note that actions of the generators and the derivation commute.
\end{rem}

\vspace{0.5cm}
\textbf{$r^{(n)}$-transformation.}

First of all, let us notice the following fact.
\begin{lem} \label{LenardProp1}
$\mathcal{L}_k \LieBrackets{- w^{\prime} - w^2}
    = \mathcal{L}_k \LieBrackets{w^{\prime} - w^2}
    - 2 \dfrac{d}{dz} 
    \brackets{\dfrac{d}{dz} + 2 w} 
    \mathcal{L}_{k - 1} \LieBrackets{w^{\prime} - w^2}, 
    \quad 1 \leq k \leq n$.
\end{lem}
\begin{proof}[Proof of Lemma \ref{LenardProp1}]
By mathematical induction and the Lenard recursion relation.
\end{proof}

Then $P_k$, $k = 1, \dots, n$, transforms as
\begin{align}
    r^{(n)} \brackets{P_k}
    &= \dfrac{1}{2^{2k - 1}} 
    \sum_{l = 1}^{n} t_l 
    \mathcal{L}_{l - \brackets{n - k}} \LieBrackets{- w^{\prime} - w^2}
    \\
    &= P_k 
    - \dfrac{1}{2^{2k}} \sum_{l = 1}^n t_l 
    \dfrac{d}{dz} \brackets{\dfrac{d}{dz} + 2 w}
    \mathcal{L}_{l - \brackets{n - k} - 1} \LieBrackets{w^{\prime} - w^2}
    \\
    &= P_k - \dfrac{1}{2^{2n - 1}} b_{2 \brackets{n - k} + 1}^{(n)}.
\end{align}

For $Q_k$, $k = 1, \dots, n - 1$, we have
\begin{align}
    r^{(n)} \brackets{Q_k}
    & = - r^{(n)} \brackets{b_{2k}^{(n)}} 
    - \sum_{l = 1}^{n - k} r^{(n)} \brackets{P_l}  r^{(n)} \brackets{Q_{l + k}}.
\end{align}

\begin{lem} \label{rTransOf_b_k^n}
$r^{(n)} \brackets{b_{k}^{(n)}}
    = - b_{k}^{(n)}$.
\end{lem}
\begin{proof}[Proof of Lemma \ref{rTransOf_b_k^n}]
It follows from the discrete symmetry $w \brackets{z, \overline{t}; - \alpha_n} \mapsto - w \brackets{z, \overline{t}; \alpha_n}$ of \eqref{nonstatPIIhier}.
\end{proof}

Therefore,
\begin{align}
    r^{(n)} \brackets{Q_k}
    & = b_{2k}^{(n)}
    - \sum_{l = 1}^{n - k} r^{(n)} \brackets{P_l}  r^{(n)} \brackets{Q_{l + k}}.
\end{align}

Action on a parameter $b_n$ is a consequence of the following
\begin{lem} \label{derivativeOf_b_1^n}
$    \dfrac{d}{dz} b_1^{(n)}
    = 2 Q_n P_n 
    - \dfrac{1}{2^{2n - 1}} Q_n b_1^{(1)}
    - \brackets{b_n - \dfrac{1}{2}}$.
\end{lem}
\begin{proof}[Proof of Lemma \ref{derivativeOf_b_1^n}]
Follows straightforward.
\end{proof}

\vspace{0.5cm}
\textbf{2. Fundamental relations.}

\vspace{0.5cm}
\textbf{$s^{(n)}$-transformation.}

To prove fundamental relations for the $s^{(n)}$-generator, we should compute its double action on $Q_n$ and $b_n$ only. So,
\begin{align}
    \brackets{s^{(n)} \brackets{Q_n}}^2
    &= s^{(n)} \brackets{Q_n - \dfrac{b_n}{P_n}}
    = s^{(n)} \brackets{Q_n}
    - \dfrac{s^{(n)}\brackets{b_n}}{s^{(n)} \brackets{P_n}}
    = Q_n, 
    \\
    \brackets{s^{(n)} \brackets{b_n}}^2
    &= - s^{(n)} \brackets{b_n}
    = b_n.
\end{align}

\vspace{0.5cm}
\textbf{$r^{(n)}$-transformation.}

Similarly, one can prove that the $r^{(n)}$-generator satisfies the fundamental relation $\brackets{r^{(n)}}^2 = 1$ for $Q_i$, $P_i$, $i = 1, \dots, n$, and $b_n.$

\vspace{0.5cm}
\textbf{3. The canonicity of new coordinates.}

The canonicity of new coordinates defined as a result of actions of generators $s^{(n)}$ and $r^{(n)}$ on origin canonical coordinates follows from the given non-zero Poisson brackets \cite{mazzocco2007hamiltonian}
\begin{align}
    \PoissonBrackets{a_{2k + 1}^{(n)}, b_{2 l + 1}^{(n)}}
    &= - b_{2 (k + l + 1)}^{(n)},
    & 
    0& \leq k \leq n, 
    &
    0& \leq l \leq n - 1,
    & 
    k& + l \leq n - 1,
    \\
    \PoissonBrackets{a_{2k + 1}^{(n)}, b_{2 l}^{(n)}}
    &= - b_{2 (k + l) + 1}^{(n)},
    & 
    0& \leq k \leq n, 
    &
    1& \leq l \leq n,
    & 
    k& + l \leq n - 1,
    \\
    \PoissonBrackets{b_{2k}^{(n)}, b_{2 l + 1}^{(n)}}
    &= a_{2 (k + l) + 1}^{(n)},
    & 
    1& \leq k \leq n, 
    &
    0& \leq l \leq n - 1,
    & 
    k& + l \leq n.
\end{align}
\end{proof}

\begin{corl}
An extended affine Weyl group $\tilde{W}^{(n)}$ associated with the $n$-th member of \text{\rm $\PIIn{n}$} has following generators
\begin{equation}
    \tilde{W}^{(n)} = \angleBrackets{s_0^{(n)}, s_1^{(n)}; \pi^{(n)}},
    \quad 
    s_0^{(n)} = s^{(n)}, \, s_1^{(n)} = r^{(n)}s^{(n)}r^{(n)}, \, \pi^{(n)} = r^{(n)},
\end{equation}
with fundamental relations $\brackets{s_i^{(n)}}^2 = 1$, $\brackets{\pi^{(n)}}^2 = 1$, $\pi^{(n)} s_i^{(n)} = s_{i+1}^{(n)} \pi^{(n)}$, $i \in \quot{\mathbb{Z}}{2 \mathbb{Z}}$.
\end{corl}

\begin{exmp}
\phantom{}

For the case $n = 1$, by Theorem 6.1 in \cite{mazzocco2007hamiltonian}, some entries of a matrix $\mathcal{A}^{(1)}$ are
\begin{equation}
    b_1^{(1)}
    = 4 P_1 + \dfrac{1}{8} Q_1^2 + z,
\end{equation}
and, by Theorem \ref{WeilConnection}, we obtain formulas \eqref{BT1genQP}.

For the case $n = 2$, some entries are
\begin{gather}
b_1^{(2)}
    = 16 P_2 - 8 P_1^2 + \dfrac{1}{16} Q_1 Q_2 + \dfrac{1}{32} P_1 Q_2^2 + z + t_1 \brackets{4 P_1 - \dfrac{t_1}{2}}, 
    \quad
    b_3^{(2)}
    = 16 P_1 + \dfrac{1}{32} Q_2^2 - 4 t_1;
    \\ 
    b_2^{(2)}
    = - Q_1 - P_1 Q_2,
\end{gather}
thus, generators of an affine Weyl group $W^{(2)}$ are given as
\begin{gather}
    s^{(2)} \brackets{Q_1}
    = Q_1,
    \quad 
    s^{(2)} \brackets{Q_2}
    = Q_2 - \dfrac{b_2}{P_2},
    \quad 
    s^{(2)} \brackets{P_1}
    = P_1,
    \quad 
    s^{(2)} \brackets{P_2}
    = P_2,
    \quad 
    s^{(2)} \brackets{b_2}
    = - b_2;
    \\
    r^{(2)} \brackets{Q_1}
    = - Q_1 - Q_2 \brackets{2 P_1 + \dfrac{1}{2^8} Q_2^2 - \dfrac{t_1}{2}},
    \quad 
    r^{(2)} \brackets{Q_2}
    = - Q_2,
    \\ 
    r^{(2)} \brackets{P_1}
    = - \brackets{P_1 + \dfrac{1}{2^8} Q_2^2 - \dfrac{t_1}{2}},
    \quad 
    r^{(2)} \brackets{P_2}
    = - P_2 + P_1^2 - \dfrac{1}{2^7} Q_1 Q_2 - \dfrac{1}{2^8} P_1 Q_2^2 - \dfrac{z}{8} - \dfrac{t_1}{16} \brackets{4 P_1 - \dfrac{t_1}{2}},
    \\ 
    r^{(2)} \brackets{b_2}
    = - b_2 - 2.
\end{gather}

The canonicity of new variables and fundamental relations can be verified by a straightforward computation.
\end{exmp}

\begin{rem}
The generator construction process is usually based on an explicit Hamiltonian system. However, in our case, this method is complicated, since we should substitute canonical coordinates derivatives into formulas \eqref{canCoordAsPolsEq} after actions of generators on them iteratively. For instance, in the case $n = 1$, $r^{(1)}$ transforms $Q_1$ to $\tilde{Q}_1 = - \frac{1}{4} w$, and then $r^{(1)} \brackets{P_1} = \tilde{P}_1$ is
\begin{align}
    \tilde{P}_1
    &= \dfrac{1}{2} \brackets{- w^{\prime} - w^2 - \dfrac{z}{2}}
    = - \dfrac{1}{8} Q_1^{\prime} 
    - \dfrac{1}{8} Q_1^2 
    - \dfrac{z}{4}
    \\
    &= - \dfrac{1}{8} \brackets{8 P_1 + \dfrac{1}{4} Q_1^2 + 2 z} 
    - \dfrac{1}{8} Q_1^2 
    - \dfrac{z}{4}
    = - \brackets{P_1 + \dfrac{1}{16} Q_2^2 + \dfrac{z}{2}}.
\end{align}
\end{rem}
\subsection{Variables of a symmetric form for the \texorpdfstring{$\PIIn{n}$}{PIIn} hierarchy}

The symmetric form of the second \Painleve equation \eqref{PII1} is constructed by an extension generator $\pi^{(1)}$ of the extended affine Weyl group $\Tilde{W^{(1)}}$. Let us consider new variables defined as
\begin{equation}
    f_0^{(1)}
    = P_1, 
    \quad 
    f_1^{(1)}
    = \pi^{(1)} \brackets{P_1}
    = - P_1 - \dfrac{1}{16} Q_1^2 - \dfrac{z}{2}.
\end{equation}

Since actions of $\tilde{W}^{(1)}$ generators commute with the derivation, the Hamiltonian system constructed by $\HIIn{1}$ \eqref{HS1} is
\begin{gather}
    \label{symFormPII1}
    \begin{cases}
        {f_0^{(1)}}^{\prime} 
        &= - \dfrac{1}{2} Q_1 f_0^{(1)} 
        + \dfrac{1}{4} \beta_0^{(1)}, \\
        {f_1^{(1)}}^{\prime} 
        &= \dfrac{1}{2} Q_1 f_1^{(1)} 
        + \dfrac{1}{4} \beta_1^{(1)}, \\
        Q_1^{\prime} 
        &= 4 \brackets{f_0^{(1)} - f_1^{(1)}},
    \end{cases}
\end{gather}
with $\beta_0^{(1)} + \beta_1^{(1)} = - 2$.

System \eqref{symFormPII1} is called \textit{the symmetric form for the second \Painleve equation} due to the fact that actions of $\pi^{(1)}$ preserve it.
\begin{table}[H]
    \centering
    \begin{tabular}{c||cc|ccc}
      & 
      $\beta_0^{(1)}$ &  
      $\beta_1^{(1)}$ & 
      $f_0^{(1)}$ & 
      $f_1^{(1)}$ &
      $Q_1$ \\
      \hline
      \hline
      $s_0^{(1)} \vsp$ &  
      $- \beta_0^{(1)}$ & 
      $2 \beta_0^{(1)} + \beta_1^{(1)}$ & 
      $f_0^{(1)}$ &
      $f_1^{(1)} + \dfrac{1}{8} Q_1 \dfrac{\beta_0^{(1)}}{f_0^{(1)}} - \dfrac{1}{16} \dfrac{{\beta_0^{(1)}}^2}{{f_0^{(1)}}^2}$ &
      $Q_1 - \dfrac{\beta_0^{(1)}}{f_0^{(1)}}$ \\
      $s_1^{(1)} \vsp$ & 
      $\beta_0^{(1)} + 2 \beta_1^{(1)}$ & 
      $- \beta_1^{(1)}$ & 
      $f_0^{(1)} - \dfrac{1}{8} Q_1 \dfrac{\beta_1^{(1)}}{f_1^{(1)}} - \dfrac{1}{16} \dfrac{{\beta_1^{(1)}}^2}{{f_1^{(1)}}^2}$ &
      $f_1^{(1)}$ &
      $Q_1 + \dfrac{\beta_1^{(1)}}{f_1^{(1)}}$ \\
      \hline
      $\pi^{(1)} \vsp$ & 
      $\beta_1^{(1)}$ &  
      $\beta_0^{(1)}$ & 
      $f_1^{(1)}$ & 
      $f_0^{(1)}$ &
      $- Q_1$
    \end{tabular} 
\caption{\Backlund transformations of the symmetric form \eqref{symFormPII1} for $\PIIn{1} \LieBrackets{ w  \brackets{z; \alpha_1}}$}
\end{table}

Note that if we consider the standard Poisson structure for variables $\brackets{Q_1, P_1, z}$, \Backlund transformations of the symmetric form are represented in the following way
\begin{align}
    s_i^{(1)} \brackets{\varphi}
    &= \exp{\brackets{\dfrac{\beta_i^{(1)}}{f_i^{(1)}} \ad_{\PoissonBrackets{,}} \brackets{f_i^{(1)}} \brackets{\varphi}}},
    \\
    \ad_{\PoissonBrackets{,}} \brackets{\varphi}
    &= \PoissonBrackets{\varphi, \cdot}.
\end{align}

This construction carries over to the case of an arbitrary member number $n$ of \eqref{nonstatPIIhier}.

\begin{thm} \label{SFthm}
A symmetric form for the $n$-th member of the non-stationary \text{\rm $\PIIn{n}$} hierarchy is defined in the following variables
\begin{align} \label{symVar}
    \begin{aligned}
    f_{n - k + 1,0}^{(n)}
    &= P_k,
    &
    f_{n - k + 1,1}^{(n)}
    &= \pi^{(1)} \brackets{P_k},
    \quad 
    1 \leq k \leq n,
    &
    \pi^{(n)} \brackets{f_{k, i}^{(n)}}
    &= f_{k, i + 1},
    & \quad
    &
    \\
    g_{n - k,0}^{(n)}
    &= Q_k,
    &
    g_{n - k,1}^{(n)}
    &= \pi^{(1)} \brackets{Q_k},
    \quad 
    1 \leq k \leq n - 1,
    &
    \pi^{(n)} \brackets{g_{k, i}^{(n)}}
    &= g_{k, i + 1},
    & \quad
    i &\in \quot{\mathbb{Z}}{2 \mathbb{Z}},
    \\
    \beta_0^{(n)}
    &= b_n,
    &
    \beta_1^{(n)}
    &= \pi^{(1)} \brackets{b_n},
    \quad 
    &
    \pi^{(n)} \brackets{\beta_i^{(n)}}
    &=\beta_{i + 1}^{(n)},
    & \quad
    &
    \end{aligned}
\end{align}
with the additional equation
\begin{equation} \label{symAdEq}
    Q_n^{\prime}
    = 4^n \brackets{f_{n,0}^{(n)} - f_{n,1}^{(n)}},
\end{equation}
and the condition $\beta_0^{(n)} + \beta_1^{(n)} = - 2$.
\end{thm}
\begin{proof}
Actions of an affine Weyl group commute with the derivation. Hence, variables \eqref{symVar} represent $2n$ equations of a Hamiltonian system in a symmetric form. To get from this system the $n$-th member of \eqref{nonstatPIIhier}, we should use additional equation \eqref{symAdEq}. The form of this equation follows from the corresponding Hamiltonian equation, i.e.
\begin{equation}
    Q_n^{\prime} 
    = 4^n \brackets{2 P_1 + \dfrac{1}{4^{2n}} Q_n^2 - \dfrac{t_{n - 1}}{2}},
    \quad
    t_0 = - z.
\end{equation}

The expression in parentheses equals $\brackets{f_{n,0}^{(n)} - f_{n,1}^{(n)}}$.
\end{proof}

\begin{corl}
\Backlund transformations of the symmetric form for the $n$-th member of the non-stationary \text{\rm $\PIIn{n}$} hierarchy can be represented as
\begin{align}
    s_i^{(n)} \brackets{\varphi}
    &= \exp{\brackets{\dfrac{\beta_i^{(n)}}{f_{1,i}^{(n)}} \ad_{\PoissonBrackets{,}} \brackets{f_{1,i}^{(n)}} \brackets{\varphi}}},
    \quad
    i \in \quot{\mathbb{Z}}{2 \mathbb{Z}},
\end{align}
where $\ad_{\PoissonBrackets{,}} \brackets{\varphi} = \PoissonBrackets{\varphi, \cdot}$.
\end{corl}
\begin{proof}
This fact follows from Theorems \ref{WeilConnection} and \ref{SFthm}.
\end{proof}

\begin{exmp}
For the case $n = 2$, variables of a symmetric form are
\begin{align}
    f_{1,0}^{(2)}
    &= P_2, 
    & 
    f_{1,1}^{(2)}
    &= \pi^{(2)} \brackets{P_2} 
    = P_1^2 - P_2 - \dfrac{1}{2^7} Q_1 Q_2 - \dfrac{1}{2^8} P_1 Q_2^2 - \dfrac{1}{8} z - \dfrac{1}{2} t_1 P_1 + \dfrac{1}{16} t_1^2
    ; 
    \\
    f_{2,0}^{(2)}
    &= P_1, 
    & 
    f_{2,1}^{(2)}
    &= \pi^{(2)} \brackets{P_1} 
    = - \brackets{P_1 + \dfrac{1}{2^8} Q_2^2 - \dfrac{t_1}{2}}
    ; 
    \label{generatorsSymFormPII2}
    \\
    g_{1,0}^{(2)}
    &= Q_1, 
    & 
    g_{1,1}^{(2)}
    &= \pi \brackets{Q_1} 
    = - 2 P_1 Q_2 - Q_1 - \dfrac{1}{2^8} Q_2^3 + \dfrac{1}{2}t_1 Q_2
    .
\end{align}

Therefore, the symmetric form for $\PIIn{2}$ is given by the following system
\begin{equation} \label{symFormPII2}
    \begin{cases}
        {f_{1,0}^{(2)}}^{\prime}
        &= - \dfrac{1}{8} Q_2 f_{1,0}^{(2)} 
        + \dfrac{1}{16} \beta_0^{(2)}, 
        \vspace{0.3cm}
        \\
        {f_{1,1}^{(2)}}^{\prime}
        &= \dfrac{1}{8} Q_2 f_{1,1}^{(2)} 
        + \dfrac{1}{16} \beta_1^{(2)}, 
        \vspace{0.3cm}
        \\
        {f_{2,0}^{(2)}}^{\prime}
        &= \dfrac{1}{8} g_{1,0}^{(2)}, 
        \vspace{0.3cm}
        \\
        {f_{2,1}^{(2)}}^{\prime}
        &= \dfrac{1}{8} g_{1,1}^{(2)}, 
        \vspace{0.3cm}
        \\
        {g_{1,0}^{(2)}}^{\prime} 
        &= - 48 {f_{2,0}^{(2)}}^2 
        + 32 f_{1,0}^{(2)}
        + 16 f_{2,0}^{(2)} t_1 
        - t_1^2 
        + 2 z , 
        \vspace{0.3cm}
        \\
        {g_{1,1}^{(2)}}^{\prime} 
        &= - 48 {f_{2,1}^{(2)}}^2 
        + 32 f_{1,1}^{(2)} 
        + 16 f_{2,1}^{(2)} t_1 
        - t_1^2 
        + 2z, 
        \vspace{0.3cm}
        \\
        Q_2^{\prime} 
        &= 16 \brackets{f_{2,0}^{(2)} - f_{2,1}^{(2)}};
    \end{cases}
\end{equation}
with $\beta_0^{(2)} + \beta_1^{(2)} = -2$.
\end{exmp}

\section{Yablonskii-Vorobiev-type polynomials} \label{YVsection}

The Yablonskii-Vorobiev polynomials concern with rational solutions of the second \Painleve equation. First of all, we will find its rational solutions by actions of \Backlund transformations compositions $T_m^{(1)} = \brackets{r^{(1)}s^{(1)}}^m$ and $S_k^{(1)} = \brackets{r^{(1)}s^{(1)}}^k s^{(1)}$ on the unique trivial solution $w^{(1)} \brackets{z; 0} = 0$ which corresponds to $\brackets{Q_1, P_1; b_1} = \brackets{0, - \frac{z}{4}; -1}$. One can see that if $m = - b_1 - k$ holds, rational solutions obtained by actions of $T_m^{(1)}$ and $S_k^{(1)}$ on $\brackets{Q_1, P_1; b_1} = \brackets{0, - \frac{z}{4}; -1}$ are the same. Therefore, use one of the operators $T_m^{(1)}$, $S_k^{(1)}$ is sufficient to generate unique rational solutions of $\PIIn{1} \LieBrackets{w\brackets{z;\alpha_1}}$.
\begin{table}[H]
    \centering
    \begin{tabular}{c||cccccc}
      & $\cdots$ & $T_{-1}^{(1)}$ & $T_{0}^{(1)}$ & $T_{1}^{(1)}$ & $T_{2}^{(1)}$ & $\cdots$ \\
      \hline
      \hline
      $Q_1 \vsp$ 
      & $\cdots$ & $\dfrac{4}{z}$ & $0$ & $- \dfrac{4}{z}$ & $- 8 \dfrac{z^3 - 2}{z \brackets{z^3 + 4}}$ & $\cdots$ \\
      $P_1 \vsp$ 
      & $\cdots$ & $- \dfrac{z^3 + 4}{4 z^2}$ & $- \dfrac{z}{4}$ & $- \dfrac{z}{4}$ & $- \dfrac{z^3 + 4}{4 z^2}$ & $\cdots$ \\
      \hline
      $b_1$ & $\cdots$ & $-3$ & $-1$ & $1$ & $3$ & $\cdots$ \\
    \end{tabular} 
\caption{Actions of $T_m^{(1)}$ on $\brackets{Q_1, P_1; b_1} = \brackets{0, - \dfrac{z}{4}; -1}$}
\label{tab:ratSolPII1}
\end{table}

Polynomials appeared in enumerators and denominators satisfy the difference-differential equation obtained in \cite{yablonskii1959rational, vorobiev1965rational}
\begin{align}
    u_{m + 1}^{\brackets{1}} u_{m - 1}^{\brackets{1}}
    &= z \brackets{u_m^{\brackets{1}}}^2 - 4  \brackets{u_m^{\brackets{1}}
    \brackets{u_m^{\brackets{1}}}^{\prime\prime}
    - \brackets{\brackets{u_m^{\brackets{1}}}^{\prime}}^2},
\end{align}  
with $u_0^{\brackets{1}}  = 1$, $u_1^{\brackets{1}} = z$, and $m \coloneqq \alpha_1 \in \mathbb{Z}$.

Some of them are given as
\begin{gather}
    u_2^{(1)} 
    = 4+z^3,
    \quad
    u_3^{(1)} 
    = -80+20 z^3+z^6,
    \quad
    u_4^{(1)} 
    = z (11200+60 z^6+z^9).
\end{gather}

Hence, rational solutions of \eqref{PII1} are expressed via Yablonskii-Vorobiev polynomials:
\begin{gather}
    w_1^{(1)} (z) 
    = \dfrac{d}{dz} \ln \brackets{\dfrac{u_{0}^{(1)}}{u_1^{(1)}}}
    = -\dfrac{1}{z},
    \quad 
    \alpha_1 = 1;
    \quad
    w_2^{(1)} (z) 
    = \dfrac{d}{dz} \ln \brackets{\dfrac{u_{1}^{(1)}}{u_2^{(1)}}}
    = \dfrac{4-2 z^3}{z^4+4 z},
    \quad 
    \alpha_1 = 2;
    \\
    w_3^{(1)} (z) 
    = \dfrac{d}{dz} \ln \brackets{\dfrac{u_{2}^{(1)}}{u_3^{(1)}}}
    = -\dfrac{3 z^2 \left(z^6+8 z^3+160\right)}{z^9+24 z^6-320},
    \quad 
    \alpha_1 = 3.
\end{gather}

We can see that these solutions coincide with rational solutions in Table \ref{tab:ratSolPII1}, i.e. $w_m^{(1)} \brackets{z} = \frac{1}{4} T_m^{(1)} \brackets{Q_1}$.

In this section we find a generalization of the Yablonskii-Volobiev polynomials and define them for an arbitrary member of the non-stationary second \Painleve hierarchy \eqref{nonstatPIIhier}.

\begin{thm} \label{YVthm}
    A difference-differential recurrent equation for Yablonskii-Vorobiev-type polynomials associated with the $n$-th member of the non-stationary \text{\rm $\PIIn{n}$} hierarchy is the following equation
    \begin{gather} \label{nonstatYB}
        u_{m + 1}^{\brackets{n}} u_{m - 1}^{\brackets{n}}
        = - 2 \brackets{u_m^{\brackets{n}}}^2 
        \sum_{l = 0}^{n} t_l 
        \mathcal{L}_l \LieBrackets{2 \dfrac{d^2}{d z^2} \brackets{\ln{u_m^{\brackets{n}}}}},
        \quad 
        t_0 = -z,
        \,\,
        t_n = 1,
    \end{gather}
    with $u_0^{\brackets{n}} = 1$ and $u_1^{\brackets{n}} = z$.
\end{thm}

\begin{proof}
The proof of Theorem \ref{YVthm} is based on Lemmas \ref{solViaYV} and \ref{propYV}, which we will establish after the theorem's proof.
    \begin{lem} \label{solViaYV}
    Rational solutions of \eqref{nonstatPIIhier} for an integer parameter $\alpha_n \coloneqq m \in \mathbb{Z}$ are unique and represented in the given forms
    \begin{equation} \label{solViaYVform}
        w_m^{(n)} \brackets{z} = \dfrac{d}{dz} \ln \brackets{\dfrac{u_{m-1}^{(n)}}{u_m^{(n)}}},
        \quad 
        w_{-m}^{(n)} \brackets{z} = - w_m^{(n)} \brackets{z},
    \end{equation}
    where $u_k^{(n)}$ are the Yablonskii-Vorobiev-type polynomials.
    \end{lem}

    \begin{lem} \label{propYV}
    Yablonskii-Vorobiev-type polynomials satisfy the following recurrent relations
    \begin{itemize}
        \item 
        $D_z u_{m + 1}^{(n)} \cdot u_{m - 1}^{(n)} = (2 m + 1 ) \brackets{u_m^{(n)}}^2$,
        
        \item
        $D_z^2 u_{m + 1}^{(n)} \cdot u_{m}^{(n)} = 0$,
    \end{itemize}
    where $D_z$ is the Hirota operator, i.e.
    $
    D_z f \brackets{z} \cdot g \brackets{z} 
        = \left. \brackets{\dfrac{d}{dz_1} - \dfrac{d}{dz_2}}\brackets{f \brackets{z_1} g \brackets{z_2}} \right|_{z_1=z_2=z}
    $.
    \end{lem}
    
    Let us consider \Backlund transformation \eqref{nonstatBT} with $\varepsilon = 1$ and transform a parameter $-m$ to $m + 1$:
    \begin{equation} \label{proof2}
        w_{m+1}^{(n)}
        = - w_m^{(n)}
        +
        \dfrac{2 m + 1}{2 \dsum_{l = 0}^{n} t_l \mathcal{L}_l \LieBrackets{- {w_m^{(n)}}^{\prime} - {w_m^{(n)}}^2}}.
    \end{equation}
    
    Then, by Lemma \ref{solViaYV}, we express $w_{m+1}^{(n)} \brackets{z}$ and $w_m^{(n)} \brackets{z}$ in terms of the Yablonskii-Vorobiev-type polynomials, using properties in Lemma \ref{propYV}:
    \begin{align} \label{proof31}
        w_{m+1}^{(n)} + w_m^{(n)}
        &= 
        \left(
        \dfrac{\brackets{u_{m}^{(n)}}^{\prime}}{u_{m}^{(n)}}
        -
        \dfrac{\brackets{u_{m+1}^{(n)}}^{\prime}}{u_{m+1}^{(n)}}
        \right)
        +
        \left(
        \dfrac{\brackets{u_{m - 1}^{(n)}}^{\prime}}{u_{m - 1}^{(n)}}
        -
        \dfrac{\brackets{u_{m}^{(n)}}^{\prime}}{u_{m}^{(n)}}
        \right)
        =
        - \dfrac{D_z u_{m + 1}^{(n)} \cdot u_{m - 1}^{(n)}}{u_{m + 1}^{(n)} u_{m - 1}^{(n)}},
    \end{align}
    \begin{align} 
        \label{proof32}
        - {w_m^{(n)}}^{\prime} - {w_m^{(n)}}^2
        &= 
        - 2 \brackets{\dfrac{\brackets{u_{m}^{(n)}}^{\prime}}{u_{m}^{(n)}}}^2
        + 2 \dfrac{\brackets{u_{m}^{(n)}}^{\prime\prime}}{u_{m}^{(n)}}
        - \dfrac{\brackets{u_{m - 1}^{(n)}}^{\prime\prime}}{u_{m - 1}^{(n)}}
        - \dfrac{\brackets{u_{m}^{(n)}}^{\prime\prime}}{u_{m}^{(n)}}
        + 2 \dfrac{\brackets{u_{m - 1}^{(n)}}^{\prime}\brackets{u_{m}^{(n)}}^{\prime}}{u_{m - 1}^{(n)}u_{m}^{(n)}}
        \\
        &=
        2 \dfrac{d^2}{dz^2} \ln \brackets{u_m^{(n)}} 
        - \dfrac{D_z^2 u_{m}^{(n)} \cdot u_{m - 1}^{(n)}}{u_{m}^{(n)} u_{m - 1}^{(n)}}
        = 2 \dfrac{d^2}{dz^2} \ln \brackets{u_m^{(n)}}.
    \end{align}
    
    Substituting \eqref{proof31}, \eqref{proof32} into \eqref{proof2}, we have the statement of our theorem.
\end{proof}

\begin{corl}
Rational solutions $u \brackets{x, t_{2n + 1}}$ of the $n$-th member of the Korteveg-de Vriez hierarchy are given by the Yablonskii-Vorobiev-type polynomials as
\begin{equation}
    u \brackets{x, t_{2n + 1}}
    = 2 \dfrac{\partial^2}{\partial x^2} \PoissonBrackets{\ln \LieBrackets{u_m^{(n)} \brackets{\dfrac{x}{\brackets{\brackets{2 n + 1} t_{2 n + 1}}^{\frac{1}{2n + 1}}}}}}.
\end{equation}
\end{corl}

\begin{proof}[Proof of Lemma \ref{solViaYV}]
According to \cite{gromak2008painleve}, it is possible to assign a meromorphic function to each member of the stationary second \Painleve hierarchy, which we will do in the non-stationary case. Let us determine a meromorphic function $W \brackets{z}$ such that $W^{\prime} \brackets{z} = w^2$ and an entire function $u \brackets{z}$ such that $u^{\prime} = - W u$. Define a function $v \brackets{z}$ such that $v \brackets{z} \coloneqq u \brackets{z} w \brackets{z}$. Then, on the one hand, 
\begin{equation}
    u^{\prime\prime} 
    = - W^{\prime} u + W^2 u.
\end{equation}

On the other hand,
\begin{align}
    \dfrac{u^{\prime\prime}}{u}
    &= \brackets{\dfrac{u^{\prime}}{u}}^{\prime}
    + \brackets{\dfrac{u^{\prime}}{u}}^2
    = - W^{\prime} 
    + \brackets{\dfrac{u^{\prime}}{u}}^2
    \\
    &= - w^2 
    + \brackets{\dfrac{u^{\prime}}{u}}^2
    = - \brackets{\dfrac{v}{u}}^2 
    + \brackets{\dfrac{u^{\prime}}{u}}^2,
    \\
    u u^{\prime\prime}
    &= \brackets{u^{\prime}}^2 - v^2.
\end{align}

Therefore, we have the following representation of \eqref{nonstatPIIhier}
\begin{equation} \label{bilinSystem}
    \begin{cases}
    u u^{\prime\prime}
    = \brackets{u^{\prime}}^2 - v^2,
    \vspace{0.2cm}
    \\
    \brackets{\dfrac{d}{dz} + 2 v u^{-1}}
    \dsum_{l = 0}^n t_l \mathcal{L}_l \LieBrackets{v^{\prime} u^{-1} - v u^{\prime} u^{-2}}
    = \alpha_n.
    \end{cases}
\end{equation}

If the initial condition $u \brackets{z_0} \neq 0$ holds, then system \eqref{bilinSystem} with initial conditions
\begin{align}
    u \brackets{z_0}
    &= u_0,
    &
    u^{\prime} \brackets{z_0}
    &= u_0^{\prime}, 
    &&&&
    \\
    v \brackets{z_0}
    &= v_0,
    &
    v^{\prime} \brackets{z_0}
    &= v_0^{\prime},
    &
    \dots
    &&
    v^{\brackets{2n - 1}} \brackets{z_0}
    &= v_0^{\brackets{2n - 1}},
\end{align}
has a unique non-singular solution. Obviously, system \eqref{bilinSystem} has a solution 
\begin{equation}
    \brackets{v, u; \alpha_n}
    = \brackets{0, e^{a z + b}; 0},
    \quad 
    a, b \in \mathbb{C}.
\end{equation}

As in \cite{gromak2008painleve}, any meromorphic solution of \eqref{nonstatPIIhier} can be represented in the form $w \brackets{z} = v \brackets{z}/u\brackets{z}$, where functions $v \brackets{z}$ and $u \brackets{z}$ are polynomials defined by system \eqref{bilinSystem} up to the factor $e^{az + b}$. If we set $p = \deg v \brackets{z}$ and $q = \deg u \brackets{z}$, then from the first equation of \eqref{bilinSystem} follows the degrees relation 
\begin{equation}
    2 p = 2 \brackets{q - 1},
    \quad 
    \Rightarrow 
    \quad 
    p = q - 1,
\end{equation}
i.e. the degree of the polynomial in the denominator is more by one of the polynomial degree in the nominator.

Hence, the Laurent expansion of the rational solution of \eqref{nonstatPIIhier} at a pole $z = z_0$ is
\begin{equation}
    w \brackets{z}
    = \dfrac{a_{-1}}{z - z_0} 
    + \varphi \brackets{z},
\end{equation}
where $\varphi \brackets{z}$ is a holomorphic function in a neighborhood of $z_0$.

On the other hand, the Laurent expansion of the rational solution of \eqref{nonstatPIIhier} at $z = \infty$ is given as
\begin{equation}
    w \brackets{z}
    = \dfrac{\alpha_n}{z} 
    + \psi \brackets{z},
\end{equation}
where $\psi \brackets{z}$ is a holomorphic function in a neighborhood of infinity.

So, we would like to apply the total sum of residues theorem to establish that $\alpha_n \in \mathbb{Z}$. 

\begin{prop}\label{CnCoeff}
Values of the parameter $a_{-1}$ are roots of the following equation
\begin{equation} \label{CnEq}
    c_n \brackets{a_{-1} - n} = 0,
\end{equation}
and equal $0$, $\pm 1$, $\dots$, $\pm n$, where 
\begin{equation} \label{cnExpr}
    c_n
    = \brackets{- 1}^n
    \binom{2 n + 1}{n + 1} 
    \prod_{k = 0}^{n - 1}
    \brackets{a_{-1} \brackets{a_{-1} + 1} - k \brackets{k + 1}},
    \quad
    n \geq 1,
\end{equation}
and $n$ is a number of the member \eqref{nonstatPIIhier}.
\end{prop}
\begin{proof}[Proof of Proposition \ref{CnCoeff}]
Without loss of generality, we set $z_0 = 0$. To find $a_{-1}$, we should substitute the Laurent expansion of $w \brackets{z}$ at $z = 0$ in \eqref{nonstatPIIhier} and collect coefficients at corresponding terms. Note that $a_{-1}$ is fully defined by the term of degree $- \brackets{2n + 1}$. Hence, we will look at the leading term in the Laurent expansion, i.e. 
\begin{equation}
    w\brackets{z}
    = \dfrac{a_{-1}}{z}.
\end{equation}

Therefore, the equation for the parameter $a_{-1}$ is
\begin{equation} \label{CnPrepEq}
    \PoissonBrackets{\brackets{\dfrac{d}{dz} + 2 \dfrac{a_{-1}}{z}} \mathcal{L}_n \LieBrackets{- \dfrac{a_{-1} \brackets{a_{-1} + 1}}{z^2}}} \cdot z^{2n + 1} = 0.
\end{equation}

To find an exact form of $\mathcal{L}_n \LieBrackets{- \dfrac{a_{-1} \brackets{a_{-1} + 1}}{z^2}}$, we suppose that this operator equals $\dfrac{c_n}{z^{2n}}$. Then, by the Lenard recursive relation, we have
\begin{align}
    \dfrac{d}{dz} 
    \mathcal{L}_{n + 1} \LieBrackets{- \dfrac{a_{-1} \brackets{a_{-1} + 1}}{z^2}}
    &=
    \dfrac{d}{dz} \dfrac{c_{n + 1}}{z^{2n + 2}}
    = - 2 \brackets{n + 1} \dfrac{c_{n + 1}}{z^{2n + 3}}
    \\
    &= \brackets{
    \dfrac{d^3}{dz^3} - 4 \dfrac{a_{-1} \brackets{a_{-1} + 1}}{z^2} \dfrac{d}{dz}
    + 4 \dfrac{a_{-1} \brackets{a_{-1} + 1}}{z^3}
    } \mathcal{L}_{n} \LieBrackets{- \dfrac{a_{-1} \brackets{a_{-1} + 1}}{z^2}}
    \\
    &=
    \brackets{
    \dfrac{d^3}{dz^3} - 4 \dfrac{a_{-1} \brackets{a_{-1} + 1}}{z^2} \dfrac{d}{dz}
    + 4 \dfrac{a_{-1} \brackets{a_{-1} + 1}}{z^3}
    } \dfrac{c_n}{z^{2n}}.
\end{align}

So, the recursive relation for $c_n$ is the following one
\begin{equation} \label{CnRecRel}
        c_{n + 1}
        = - \dfrac{2 \brackets{2 n + 1} \PoissonBrackets{a_{-1} \brackets{a_{-1} + 1} - n \brackets{n + 1}}}{n + 1} c_n,
        \quad 
        c_0 
        = \dfrac{1}{2},
\end{equation}
by which we are able to find general form \eqref{cnExpr} of coefficients $c_n$. 

If we substitute $\mathcal{L}_n \LieBrackets{- \dfrac{a_{-1} \brackets{a_{-1} + 1}}{z^2}} = \dfrac{c_n}{z^{2n}}$ into \eqref{CnPrepEq}, we obtain equation \eqref{CnEq}. Roots of $c_n$ are $k$ and $\brackets{- k - 1}$, where $k = 0, \dots, n - 1$.
\end{proof}

Since Proposition \ref{CnCoeff} holds, any rational solution of \eqref{nonstatPIIhier} can be represented as
\begin{equation}
    w \brackets{z}
    = \sum_{k = 0}^{n} k \brackets{\dfrac{v_k^{\prime}}{v_k} - \dfrac{u_k^{\prime}}{u_k}},
    \quad 
    k \coloneqq a_{-1}.
\end{equation}

The uniqueness follows from actions of \Backlund transformations on the \textbf{unique} trivial solution $\brackets{w \brackets{z}; \alpha_n} = \brackets{0; 0}$, and we are done.
\end{proof}

\begin{proof}[Proof of Lemma \ref{propYV}]
\phantom{}

\begin{itemize}
    \item 
    Recall that
    \begin{equation}
        \varepsilon w^{\prime} \brackets{z} 
        - w^2 \brackets{z}
        = \varepsilon \Tilde{w}^{\prime} \brackets{z}
        - \Tilde{w}^2 \brackets{z}, 
        \quad 
        \varepsilon = \pm 1,
    \end{equation}
    under the transformation $w \brackets{z} = \dfrac{1}{2} \dfrac{d}{dz} \ln \psi^{\prime} \brackets{z}$ (see Section \ref{BTSection}). Hence, the map $- m \mapsto m + 1$ corresponds the following identity 
    \begin{equation}
        - w_m^{\prime} \brackets{z} - w_m^2 \brackets{z}
        = w_{m + 1}^{\prime} \brackets{z} - w_{m + 1}^2 \brackets{z},
    \end{equation}
    or, equivalently,
    \begin{equation} \label{rel1}
        w_{m + 1}^{\prime} \brackets{z} 
        +  w_m^{\prime} \brackets{z}
        = 
        \brackets{w_{m+1} \brackets{z} + w_m \brackets{z}}
        \brackets{w_{m+1} \brackets{z} - w_m \brackets{z}}.
    \end{equation}
    
    Since $w_m \brackets{z} = \dfrac{d}{dz} \ln \dfrac{u_{m - 1}^{(n)}}{u_{m}^{(n)}} = \dfrac{1}{2} \dfrac{d}{dz} \ln \psi_m^{\prime}$, we can rewrite relation \eqref{rel1} in the following way
    \begin{align}
        2 \dfrac{d^2}{dz^2} 
        \ln 
        \psi_{m+1}^{\prime} \psi_m^{\prime}
        &= \dfrac{d}{dz} \ln 
        \psi_{m+1}^{\prime} \psi_m^{\prime}
        \cdot 
        \dfrac{d}{dz} \ln 
        \dfrac{\psi_{m+1}^{\prime}}{\psi_{m}^{\prime}},
        \\
        \dfrac{d^2}{dz^2}
        \ln 
        \dfrac{u_{m - 1}^{(n)}}{u_{m + 1}^{(n)}}
        &= \dfrac{d}{dz}
        \ln 
        \dfrac{u_{m - 1}^{(n)}}{u_{m + 1}^{(n)}}
        \cdot
        \dfrac{d}{dz}
        \ln 
        \dfrac{\brackets{u_{m}^{(n)}}^2}{u_{m - 1}^{(n)} u_{m + 1}^{(n)}}.
    \end{align}
    
    After integrating and choosing a convenient constant of integration, we have
    \begin{equation}
        \dfrac{d}{dz} \ln \dfrac{u_{m - 1}^{(n)}}{u_{m + 1}^{(n)}}
        = \brackets{2 m + 1} 
        \dfrac{\brackets{u_{m}^{(n)}}^2}{u_{m - 1}^{(n)} u_{m + 1}^{(n)}},
    \end{equation}
    or, equivalently,
    \begin{equation}
        D_z u_{m + 1}^{(n)} \cdot u_{m - 1}^{(n)}
        = \brackets{2 m + 1} \brackets{u_{m}^{(n)}}^2.
    \end{equation}
    
    \item
    Let us compare relations
    \begin{equation}
        w \brackets{z}
        = \dfrac{v \brackets{z}}{u \brackets{z}},
        \quad
        w_m \brackets{z} 
        = \dfrac{d}{dz} \ln \dfrac{u_{m-1}^{(n)}}{u_m^{(n)}}.
    \end{equation}
    
    If we set $v \brackets{z} = \brackets{u_{m-1}^{(n)}}^{\prime} u_m^{(n)} - u_{m-1}^{(n)} \brackets{u_{m}^{(n)}}^{\prime}$, $u \brackets{z} = u_{m-1}^{(n)} u_{m}^{(n)}$ and substitute them into the first equation of system \eqref{bilinSystem}, we have
    \begin{align}
        \brackets{u_{m-1}^{(n)}}^{\prime\prime} u_{m}^{(n)}
        + 2 \brackets{u_{m-1}^{(n)}}^{\prime} \brackets{u_{m}^{(n)}}^{\prime}
        + u_{m-1}^{(n)} \brackets{u_{m}^{(n)}}^{\prime\prime}
        &= 4 \brackets{u_{m-1}^{(n)}}^{\prime} \brackets{u_{m}^{(n)}}^{\prime},
        \\
        D_z^2 u_{m-1}^{(n)} \cdot u_{m}^{(n)} 
        &= 0.
    \end{align}
\end{itemize}
\end{proof}

\begin{corl}
A relation between the Yablonskii-Vorobiev-type polynomials and compositions of \Backlund transformation is given as
\begin{equation}
    \dfrac{d}{dz} \ln \brackets{\dfrac{u_{m - 1}^{(n)}}{u_m^{(n)}}}
    = \frac{1}{4^n} T_m^{(n)} \brackets{Q_n},
    \quad
    T_m^{(n)}
    = \brackets{r^{(n)}s^{(n)}}^m.
\end{equation}
\end{corl}

\begin{exmp}
Let us consider the case $n = 2$.
A trivial rational solution of $\PIIn{2} \LieBrackets{w \brackets{z; \alpha_2}}$ in terms of canonical coordinates is
\begin{equation} \label{rationalSolutionPII2}
    \brackets{Q_1, P_1, Q_2, P_2; b_2} 
    = \brackets{0, \dfrac{t_1}{4}, 0, - \dfrac{z}{4^2}; -1}.
\end{equation}

As in the case $n = 1$, we will consider actions of $T_m^{(2)} = \brackets{r^{(2)} s^{(2)}}^m$ on this rational solution \eqref{rationalSolutionPII2}. These actions are in Table \ref{tab:ratSolPII2}.
\begin{table}[H]
\centering
    \begin{tabular}{c||ccccccc}
      & $\cdots$
      & $T_{-2}^{(2)}$ 
      & $T_{-1}^{(2)}$ 
      & $T_{0}^{(2)}$
      & $T_{1}^{(2)}$
      & $T_{2}^{(2)}$
      & $\cdots$
      \\
      \hline
      \hline
      $Q_1 \vsp$
      & $\cdots$
      & $\dfrac{48 \left(16 t_1^2-28 t_1 z^3+z^6\right)}{\left(4 t_1+z^3\right)^3}$ 
      & $\dfrac{16}{z^3}$ 
      & $0$
      & $0$
      & $\dfrac{16}{z^3}$
      & $\cdots$
      \\
      $P_1 \vsp$
      & $\cdots$
      & $\dfrac{16 t_1^3+8 t_1^2 z^3+t_1 z^6+96 t_1 z-12 z^4}{4 \left(4 t_1+z^3\right)^2}$ 
      & $\dfrac{t_1 z^2-4}{4 z^2}$ 
      & $\dfrac{t_1}{4}$
      & $\dfrac{t_1}{4}$
      & $\dfrac{t_1 z^2 - 4}{4 z^2}$
      & $\cdots$
      \\
      $Q_2 \vsp$
      & $\cdots$
      & $-\dfrac{32 \left(2 t_1-z^3\right)}{z \left(4 t_1+z^3\right)}$ 
      & $\dfrac{16}{z}$ 
      & $0$
      & $- \dfrac{16}{z}$
      & $\dfrac{32 \brackets{2 t_1 - z^3}}{z \brackets{4 t_1 + z^3}}$
      & $\cdots$
      \\ 
      $P_2 \vsp$
      & $\cdots$
      & $\dfrac{z \left(80 t_1^2-20 t_1 z^3-z^6+144 z\right)}{16 \left(4 t_1+z^3\right)^2}$ 
      & $-\dfrac{4 t_1+z^3}{16 z^2}$ 
      & $- \dfrac{z}{16}$
      & $- \dfrac{z}{16}$
      & $- \dfrac{4 t_1 + z^3}{16 z^2}$
      & $\cdots$
      \\
      \hline 
      $b_2$
      & $\cdots$
      & $-5$ 
      & $-3$ 
      & $-1$ 
      & $1$
      & $3$
      & $\cdots$
    \end{tabular}
    \caption{Actions of $T_m^{(2)}$ on $\brackets{Q_1, P_1, Q_2, P_2; b_2} = \brackets{0, \dfrac{t_1}{4}, 0, - \dfrac{z}{4^2}; -1}$}
    \label{tab:ratSolPII2}
\end{table}

According to Theorem \ref{nonstatYB}, the recurrent equation for Yablonskii-Vorobiev-type polynomials is
\begin{align} \label{nonstatYB2}
    \begin{aligned}
    u_{m + 1}^{\brackets{2}} u_{m - 1}^{\brackets{2}}
    = z \brackets{u_m^{\brackets{2}}}^2
    &- 4 \brackets{u_m^{\brackets{2}}
    \brackets{u_m^{\brackets{2}}}^{\prime\prime\prime\prime}
    - 4 \brackets{u_m^{\brackets{2}}}^{\prime} \brackets{u_m^{\brackets{2}}}^{\prime\prime\prime}
    + 3 \brackets{\brackets{u_m^{\brackets{2}}}^{\prime\prime}}^2}
    \\
    &  \quad 
    - 4 t_1  \brackets{u_m^{\brackets{2}}
    \brackets{u_m^{\brackets{2}}}^{\prime\prime}
    - \brackets{\brackets{u_m^{\brackets{2}}}^{\prime}}^2},
    \quad 
    u_0^{\brackets{2}} = 1, 
    \quad
    u_1^{\brackets{2}} = z.
    \end{aligned}
\end{align}

Using recurrent difference-differential equation \eqref{nonstatYB2}, we can find Yablonskii-Vorobiev-type polynomials depending on a parameter $t_1$. Some of them are given as
\begin{gather}
    u_2^{(2)} 
    = 4 t_1+z^3,
    \quad
    u_3^{(2)} 
    = -80 t_1^2+20 t_1 z^3+z^6-144 z,
    \\
    u_4^{(2)} 
    = 11200 t_1^3 z+60 t_1 \left(z^5+336\right) z^2+z^{10}-1008 z^5-48384.
\end{gather}

Rational solutions of $\PIIn{2}$ are expressed via Yablonskii-Vorobiev-type polynomials by Lemma \ref{solViaYV}:
\begin{gather}
    w_1^{(2)} (z) 
    = \dfrac{d}{dz} \ln \brackets{\dfrac{u_{0}^{(2)}}{u_1^{(2)}}}
    = -\dfrac{1}{z},
    \quad 
    \alpha_2 = 1;
    \quad
    w_2^{(2)} (z) 
    = \dfrac{d}{dz} \ln \brackets{\dfrac{u_{1}^{(2)}}{u_2^{(2)}}}
    = \dfrac{4 t_1-2 z^3}{4 t_1 z+z^4},
    \quad 
    \alpha_2 = 2,
    \\
    w_3^{(2)} (z) 
    = \dfrac{d}{dz} \ln \brackets{\dfrac{u_{2}^{(2)}}{u_3^{(2)}}}
    = \dfrac{3 \left(160 t_1^2 z^2+8 t_1 \left(z^5-24\right)+\left(z^5+96\right) z^3\right)}{\left(4 t_1+z^3\right) \left(80 t_1^2-20 t_1 z^3-z^6+144 z\right)},
    \quad 
    \alpha_2 = 3.
\end{gather}

These rational solutions coincide with rational solutions in Table \ref{tab:ratSolPII2}, i.e.
\begin{equation}
    w_m^{(2)} \brackets{z} = \dfrac{1}{16} T_m^{(2)} \brackets{Q_2}.
\end{equation}

\begin{rem}
We have observed that for big values of parameters $|t_i|$ shapes of plots of roots "reduce" to plots of some young series of Yablonskii-Vorobiev-type polynomials (see Appendix \ref{YBplotsSection}).
\end{rem}
\end{exmp}

\section{Polynomial \texorpdfstring{$\tau$}{tau}-functions and their determinant representation} \label{TFsection}

One of the remarkable facts in the \Painleve theory is that rational solutions of the second \Painleve equation have two determinant representations. This means that there exists polynomial $\tau$-functions concerned with the Yablonskii-Vorobiev polynomials. We will work with the Jacobi-Trudi determinant representation.

Commonly known \cite{kajiwara1996determinant} that rational solutions of \eqref{PII1} have the Jacobi-Trudi determinant representation via the polynomial $\tau$-functions of the form
\begin{equation} \label{tau1Function}
    \tau_m^{(1)} \brackets{z}
    = 
    \begin{vmatrix}
    h_{m}^{(1)} \brackets{z}
    &
    h_{m + 1}^{(1)} \brackets{z}
    &
    \dots 
    &
    h_{2m - 1}^{(1)} \brackets{z}
    \vspace{0.2cm}
    \\
    h_{m - 2}^{(1)} \brackets{z}
    &
    h_{m - 1}^{(1)} \brackets{z}
    &
    \dots 
    &
    h_{2m - 3}^{(1)} \brackets{z}
    \vspace{0.2cm}
    \\
    \vdots
    &
    \vdots
    &
    \ddots
    &
    \vdots
    \vspace{0.2cm}
    \\
    h_{- m + 2}^{(1)} \brackets{z}
    &
    h_{- m + 3}^{(1)} \brackets{z}
    &
    \dots 
    &
    h_{1}^{(1)} \brackets{z}
    \end{vmatrix},
\end{equation}
where $h_k^{(1)} \brackets{z}$ ($h_k^{(1)} \brackets{z} = 0$, $k < 0$) are defined by a generating function in $\lambda$
\begin{align}
    \sum_{k = 0}^{\infty} 
    h_k^{(1)} \brackets{z} \lambda^k
    &= \exp{\brackets{z \lambda - \dfrac{4}{3} \lambda^3}}
    = \sum_{k = 0}^{\infty}
    \dfrac{1}{k!} \brackets{z \lambda - \dfrac{4}{3} \lambda^3}^k.
\end{align}

The first few polynomials $h_k^{(1)} \brackets{z}$ are
\begin{gather}
    h_0^{(1)} \brackets{z}
    = 1,
    \quad
    h_1^{(1)} \brackets{z}
    = z,
    \quad
    h_2^{(1)} \brackets{z}
    = \dfrac{1}{2!} z^2,
    \quad
    h_3^{(1)} \brackets{z}
    = \dfrac{1}{3!} \left(z-2\right) \left(z^2+2 z+4\right),
    \\
    h_4^{(1)} \brackets{z}
    = \dfrac{1}{4!} z \left(z^3-32\right),
    \quad
    h_5^{(1)} \brackets{z}
    = \dfrac{1}{5!} z^2 \left(z^3-80\right),
    \quad
    h_6^{(1)} \brackets{z}
    = \dfrac{1}{6!} \brackets{z^6-160 z^3+640};
\end{gather}
and polynomial $\tau$-functions are
\begin{gather}
    \tau_1^{(1)} \brackets{z}
    = 
    \begin{vmatrix}
    h_1^{(1)} \brackets{z}
    \end{vmatrix}
    = z
    = u_1^{(1)} \brackets{z},
    \quad
    \tau_2^{(1)} \brackets{z}
    = 
    \begin{vmatrix}
    h_2^{(1)} \brackets{z} & 
    h_3^{(1)} \brackets{z}
    \vspace{0.2cm}
    \\
    h_0^{(1)} \brackets{z} &
    h_1^{(1)} \brackets{z}
    \end{vmatrix}
    = \dfrac{1}{3} \brackets{z^3 + 4}
    = \dfrac{1}{3} u_2^{(1)} \brackets{z},
    \\
    \tau_3^{(1)} \brackets{z}
    = \dfrac{1}{45} \brackets{z^6 + 20 z^3 - 80}
    = \dfrac{1}{45} u_3^{(1)} \brackets{z},
    \quad
    \tau_4^{(1)} \brackets{z}
    = \dfrac{1}{4725} z \left( z^9 + 60 z^6 + 11200 \right)
    = \dfrac{1}{4725} u_4^{(1)} \brackets{z}.
\end{gather}

One can see that exists a relation between polynomial $\tau$-functions and the Yablonskii-Vorobiev polynomials:
$
        \tau_m^{(1)} \brackets{z}
        = c_m u_m^{(1)} \brackets{z},
        \, 
        c_m = \prod_{k = 1}^m \brackets{2 k + 1}^{k - m}
$. Hence, some rational solutions of \eqref{PII1} are
\begin{gather}
    w_1^{(1)} (z) 
    = \dfrac{d}{dz} \ln \brackets{\dfrac{\tau_{0}^{(1)}}{\tau_1^{(1)}}}
    = -\dfrac{1}{z},
    \quad 
    \alpha_2 = 1;
    \quad
    w_2^{(1)} (z) 
    = \dfrac{d}{dz} \ln \brackets{\dfrac{\tau_{1}^{(1)}}{\tau_2^{(1)}}}
    = \dfrac{4-2 z^3}{z^4+4 z},
    \quad 
    \alpha_2 = 2;
    \\
    w_3^{(1)} (z) 
    = \dfrac{d}{dz} \ln \brackets{\dfrac{\tau_{2}^{(1)}}{\tau_3^{(1)}}}
    = -\dfrac{3 z^2 \left(z^6+8 z^3+160\right)}{z^9+24 z^6-320},
    \quad 
    \alpha_2 = 3.
\end{gather}

Below we found out a generalization of polynomial $\tau$-functions to an arbitrary number $n$ in \eqref{nonstatPIIhier}.

\begin{thm} \label{tauFuncThm}
Rational solutions of the $n$-th member of the non-stationary \text{\rm $\PIIn{n}$} hierarchy \eqref{nonstatPIIhier} are given as
\begin{equation} \label{solViaTauFunc}
    w_m^{(n)} \brackets{z} 
    = \dfrac{d}{dz} \PoissonBrackets{\ln \LieBrackets{\dfrac{\tau_{m - 1}^{(n)}}{\tau_m^{(n)}}}},
    \quad
    w_{- m}^{(n)} \brackets{z} 
    = - w_m^{(n)} \brackets{z},
    \quad 
    m \geq 1,
\end{equation}
where polynomial $\tau$-functions $\tau_m^{(n)} \brackets{z}$ are $m \times m$ determinants
\begin{equation} \label{tauFunction}
    \tau_m^{(n)} \brackets{z}
    = 
    \begin{vmatrix}
    h_{m}^{(n)} \brackets{z}
    &
    h_{m + 1}^{(n)} \brackets{z}
    &
    \dots 
    &
    h_{2m - 1}^{(n)} \brackets{z}
    \vspace{0.2cm}
    \\
    h_{m - 2}^{(n)} \brackets{z}
    &
    h_{m - 1}^{(n)} \brackets{z}
    &
    \dots 
    &
    h_{2m - 3}^{(n)} \brackets{z}
    \vspace{0.2cm}
    \\
    \vdots
    &
    \vdots
    &
    \ddots
    &
    \vdots
    \vspace{0.2cm}
    \\
    h_{- m + 2}^{(n)} \brackets{z}
    &
    h_{- m + 3}^{(n)} \brackets{z}
    &
    \dots 
    &
    h_{1}^{(n)} \brackets{z}
    \end{vmatrix},
\end{equation}
which entries are the complete homogeneous polynomials $h_k^{(n)} \brackets{z}$, $h_k^{(n)} \brackets{z} = 0$ for $k < 0$ defined by a generating function in a formal parameter $\lambda$
\begin{equation} \label{pkPolynomials}
    \sum_{k = 0}^{\infty} 
    h_k^{(n)} \brackets{z} \lambda^k
    = \exp{\brackets{- \sum_{l = 0}^n \dfrac{4^l}{2 l + 1} t_{l} \lambda^{2 l + 1}}}, 
    \quad 
    t_0 = - z,
    \quad 
    t_n = 1.
\end{equation}
\end{thm}
\begin{proof}
It is well-known that the polynomial $\tau$-function of the Kadomtsev–Petviashvili hierarchy is expressed via the Schur polynomials \cite{jimbo1983solitons, kac2019polynomial}. Let us define the complete homogeneous polynomials through a generating function
\begin{equation}
    \sum_{k \geq 0} h_k \brackets{x} \lambda^k
    = \prod_{i = 1}^n \dfrac{1}{1 - x_i \lambda}
    = \exp {\brackets{\dsum_{k = 1}^n \dfrac{x_k}{k} \lambda^k}},
    \quad
    x = \brackets{x_1, \dots, x_n}.
\end{equation}

The statement of Theorem \ref{tauFunction} can be established by the following substitution 
\begin{equation}
    x_{2k + 1} = - t_{k},
    \quad
    x_{2k} = 0,
\end{equation}
and the formal parameter rescaling $\lambda \mapsto 2^{2/3} \lambda$.
\end{proof}

\begin{corl}
    A relation between the polynomial $\tau$-functions and the Yabloskii-Vorobiev-type polynomials does not depend on the member number $n$ of \eqref{nonstatPIIhier} and is given as
    \begin{equation} \label{tauFuncViaYBpol}
        \tau_m^{(n)} \brackets{z}
        = c_m u_m^{(n)} \brackets{z},
        \quad 
        c_m = \prod_{k = 1}^m \brackets{2 k + 1}^{k - m}.
    \end{equation}
\end{corl}

\begin{exmp}
Let us apply Theorem \ref{tauFuncThm} to the case $n = 2$. The generating function for the polynomials $h_k^{(2)}$ is
\begin{align}
    \sum_{k = 0}^{\infty} 
    h_k^{(2)} \brackets{z} \lambda^k
    = \exp{\brackets{z \lambda - \dfrac{4}{3} t_1 \lambda^3 - \dfrac{16}{5} \lambda^5}}
    = \sum_{k = 0}^{\infty}
    \dfrac{1}{k!} \brackets{z \lambda - \dfrac{4}{3} t_1 \lambda^3 - \dfrac{16}{5} \lambda^5}^k.
\end{align}

Hence, the first several polynomials $h_k^{(2)} \brackets{z}$ are
\begin{gather}
    h_0^{(2)} \brackets{z}
    = 1,
    \quad
    h_1^{(2)} \brackets{z}
    = z,
    \quad
    h_2^{(2)} \brackets{z}
    = \dfrac{1}{2!} z^2,
    \quad 
    h_3^{(2)} \brackets{z}
    = \dfrac{1}{3!} \left(z^3-8 t_1\right),
    \\
    h_4^{(2)} \brackets{z}
    = \dfrac{1}{4!} \left(z^4-32 t_1 z\right),
    \quad
    h_5^{(2)} \brackets{z}
    = \dfrac{1}{5!} \left(-80 t_1 z^2+z^5-384\right),
    \quad
    h_6^{(2)} \brackets{z}
    = \dfrac{1}{6!} \left(640 t_1^2-160 t_1 z^3+z^6-2304 z\right);
\end{gather}
therefore, $\tau$-functions related to polynomials $h_k^{(2)} \brackets{z}$ become
\begin{gather}
    \tau_2^{(2)} \brackets{z}
    = \dfrac{1}{3} \brackets{4 t_1+z^3}
    = \dfrac{1}{3} u_2^{(2)} \brackets{z},
    \quad
    \tau_3^{(2)} \brackets{z}
    = \dfrac{1}{45} \brackets{-80 t_1^2+20 t_1 z^3-144 z +z^6}
    = \dfrac{1}{45} u_3^{(2)} \brackets{z},
    \\
    \tau_4^{(2)} \brackets{z}
    = \dfrac{1}{4725} \brackets{11200 t_1^3 z+60 t_1 \left(z^5+336\right) z^2+z^{10}-1008 z^5-48384}
    = \dfrac{1}{4725} u_4^{(2)} \brackets{z}.
\end{gather}

The relation \eqref{tauFuncViaYBpol} between polynomial $\tau$-functions and Yablonskii-Vorobiev-type polynomials appears in formulas above.

Rational solutions constructed by \eqref{solViaTauFunc} coincide with rational solutions obtained by \eqref{solViaYVform}:
\begin{gather}
    w_1^{(2)} (z) 
    = \dfrac{d}{dz} \ln \brackets{\dfrac{\tau_{0}^{(2)}}{\tau_1^{(2)}}}
    = -\dfrac{1}{z},
    \quad 
    \alpha_2 = 1;
    \quad
    w_2^{(2)} (z) 
    = \dfrac{d}{dz} \ln \brackets{\dfrac{\tau_{1}^{(2)}}{\tau_2^{(2)}}}
    = \dfrac{4 t_1-2 z^3}{4 t_1 z+z^4},
    \quad 
    \alpha_2 = 2;
    \\
    w_3^{(2)} (z) 
    = \dfrac{d}{dz} \ln \brackets{\dfrac{\tau_{2}^{(2)}}{\tau_3^{(2)}}}
    = \dfrac{3 \left(160 t_1^2 z^2+8 t_1 \left(z^5-24\right)+\left(z^5+96\right) z^3\right)}{\left(4 t_1+z^3\right) \left(80 t_1^2-20 t_1 z^3-z^6+144 z\right)},
    \quad 
    \alpha_2 = 3.
\end{gather}
\end{exmp}

\begin{rem}
Besides the Jacobi-Trudi determinant representation there exists the Hankel determinant representation of the polynomial $\tau$-function $\tau_m^{(1)}$ \cite{kajiwara1996determinant}. The latter appears since its relation with the expansion of the Airy function at infinity \cite{iwasaki2002generating}. 

Let $q_k^{(1)}$ be polynomials defined by the following recurrence relation
\begin{gather}
    q_k^{(1)} \brackets{z}
    = \brackets{- 4}^{1/3} \dfrac{d q_{k - 1} \brackets{z}}{dz}
    + \sum_{l = 0}^{k - 2} 
    q_l \brackets{z} q_{k - 2 - l} \brackets{z},
    \quad 
    q_0^{(1)} \brackets{z}
    = 1,
    \,\,
    q_1^{(1)} \brackets{z}
    = \dfrac{z}{(- 4)^{1/3}}.
\end{gather}

Then $\tilde{\tau}_m^{(1)} \brackets{z}$ be the $m \times m$ determinant 
\begin{equation}
    \tilde{\tau}_m^{(1)} \brackets{z}
    = 
    \begin{vmatrix}
    q_0^{(1)} \brackets{z}
    &
    q_1^{(1)} \brackets{z}
    &
    \dots 
    &
    q_{m - 1}^{(1)} \brackets{z}
    \vspace{0.2cm}
    \\
    q_1^{(1)} \brackets{z}
    &
    q_2^{(1)} \brackets{z}
    &
    \dots 
    &
    q_m^{(1)} \brackets{z}
    \vspace{0.2cm}
    \\
    \vdots
    &
    \vdots
    &
    \ddots
    &
    \vdots
    \vspace{0.2cm}
    \\
    q_{m - 1}^{(1)} \brackets{z}
    &
    q_{m}^{(1)} \brackets{z}
    &
    \dots 
    &
    q_{2m - 2}^{(1)} \brackets{z}
    \end{vmatrix}
\end{equation}
for $m \geq 1$.

\begin{exmp} Some members of $q_k^{(1)}$ are
\begin{align}
    q_2^{(1)} \brackets{z}
    &= -\frac{1}{2} \sqrt[3]{-\frac{1}{2}} z^2,
    &
    q_3^{(1)} \brackets{z}
    &= -2 (-1)^{2/3} \sqrt[3]{2} z,
    &
    q_4^{(1)} \brackets{z}
    &= 5-\frac{z^3}{2},
    \\
    q_5^{(1)} \brackets{z}
    &= -4 \sqrt[3]{-1} 2^{2/3} z^2,
    &
    q_6^{(1)} \brackets{z}
    &= \frac{5}{4} \left(-\frac{1}{2}\right)^{2/3} z \left(z^3-40\right),
    &
    q_7^{(1)} \brackets{z}
    &= 60-16 z^3.
\end{align}
\end{exmp}

Polynomials $q_k^{(1)} \brackets{z}$ are connected with the Airy function \cite{iwasaki2002generating}. Let $\theta \brackets{z, \lambda}$ be an entire function of two variables defined by
\begin{equation}
    \theta \brackets{z, \lambda}
    = \exp{\brackets{\dfrac{2}{3} \lambda^3}}
    \Ai{\lambda^2 - \dfrac{z}{\brackets{- 4}^{1/3}}}. 
\end{equation}
Then the following asymptotic expansion as $\lambda \to \infty$ in any subsector of the sector $\abs{\arg \lambda} < \pi/2$ is the generating function for the sequence $q_k^{(1)} \brackets{z}$
\begin{equation}
    \dfrac{\partial}{\partial \lambda} \ln{\theta \brackets{z, \lambda}}
    \sim
    \sum_{k \geq 0} q_k^{(1)} \brackets{z} \brackets{- 2 \lambda}^k.
\end{equation}
\end{rem}

\appendix
\renewcommand{\thesection}{A}
\section{} \label{Aappendix}
\subsection{Root systems and Weyl groups}
The contents of this subsection refers to the book \cite{kac1990infinite}.

Let us consider a Cartan matrix
\begin{gather}
    A = \brackets{a_{ij}}_{i, j \in I},
    \quad
    a_{ii} = 2, 
    \quad
    a_{ij} \in \mathbb{Z}_{\leq 0};
    \\ 
    a_{ij} = 0 
    \Leftrightarrow 
    a_{ji} = 0
    \quad
    \brackets{i \neq j}.
\end{gather}

\begin{defn}
    A \textit{Weyl group} $W \brackets{A}$ associated with a Cartan matrix $A$ is given by generators $s_i$, $i \in I$, and fundamental relations
    \begin{align}
        s_i^2 
        &= 1
        &
        \text{ for all }
        i \in I,
        & &
        \\
        s_i s_j 
        &= s_j s_i
        &
        \text{ if }
        &
        &
        \brackets{a_{ij}, a_{ji}} 
        &= \brackets{0, 0},
        \\
        s_i s_j s_i
        &= s_j s_i s_j
        &
        \text{ if } 
        &
        &
        \brackets{a_{ij}, a_{ji}} 
        &= \brackets{-1, -1},
        \\
        s_i s_j s_i s_j
        &= s_j s_i s_j s_i
        &
        \text{ if }
        &
        &
        \brackets{a_{ij}, a_{ji}} 
        &= \brackets{-1, -2},
        \\
        s_i s_j s_i s_j s_i s_j
        &= s_j s_i s_j s_i s_j s_i
        &
        \text{ if }
        &
        &
        \brackets{a_{ij}, a_{ji}} 
        &= \brackets{-1, -3}.
    \end{align}
\end{defn}

\begin{rem}
    The Weyl group can be realized as a group of reflecting action on a vector space.
\end{rem}

For further, we will consider only \textit{symmetrizable} cases, i.e. when matrix elements $a_{ij}$ are realized by the inner product $\angleBrackets{\cdot, \cdot} : V \times V \to \mathbb{Q}$ as
\begin{equation}
    a_{ij} = \angleBrackets{\alpha_i^*, \alpha_j},
    \quad
    \alpha_i^* = \dfrac{2 \alpha_i}{\angleBrackets{\alpha_i, \alpha_i}},
    \quad
    \angleBrackets{\alpha_i, \alpha_i} \neq 0,
\end{equation}
where $\alpha_i$ and $\alpha_i^*$ are called \textit{the simple roots} and \textit{the simple coroots}. 

For each element $\alpha \in V$ we define the reflection $r_{\alpha} : V \to V$ of $\lambda \in V$ with respect to $\alpha \in V$ by the rule
\begin{equation}
    r_{\alpha} \brackets{\lambda} 
    = 
    \lambda - 2 \dfrac{\angleBrackets{\alpha, \lambda}}{\angleBrackets{\alpha, \alpha}} \alpha
    = \lambda - \angleBrackets{\alpha^*, \lambda} \alpha.
\end{equation}

Each system of roots is associated with \textit{a Dynkin diagram} which can be constructed by folllowing rules
\begin{table}[H]
    \centering
    \begin{tabular}{c|c}
         $\brackets{a_{ij}, a_{ji}}$
         &  
         $\alpha_i
         \qquad
         \alpha_j$
         \\
         \hline
         $\brackets{0, 0}$
         & 
         $\circ \qquad \circ$
         \\
         $\brackets{-1, -1}$
         & 
         $\circ \text{\-------} \circ$
         \\
         $\brackets{-2, -2}$
         & 
         $\circ \! \Longleftrightarrow \!
         \circ$
    \end{tabular}
\end{table}

\begin{defn}
    A permutation $\sigma$ of index set $I$ such that $a_{\sigma\brackets{i}\sigma\brackets{j}}= a_{ij}$ is \textit{a Dynkin diagram automorphism} $G$.
\end{defn}

\begin{defn}
    An \textit{extented Weyl group} $\tilde{W} \brackets{A}$ is generated by the Weyl group $W \brackets{A}$ together with an additional generator $\pi_{\sigma}$, where $\sigma \in G$, such that $\pi_{\sigma} s_i = s_i \pi_{\sigma}$.
\end{defn}

\begin{rem}
    The Cartan matrices are classified into three types: \textit{finite} type, \textit{affine} type, \textit{indefinite} type. Hence, the root system and Weyl groups are classified into the same types.
\end{rem}
\subsection{Plots of roots of Yablonskii-Vorobiev-type polynomials} 
\label{YBplotsSection}

\phantom{hh}
\begin{figure}[H]
\centering
\begin{minipage}[c]{0.18\linewidth}
\centering
    \includegraphics[width=\textwidth]{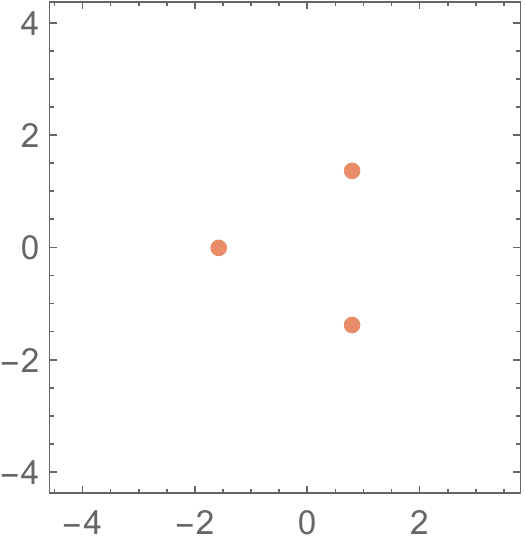}
    \text{\small $m = 2$}
\end{minipage}
\begin{minipage}[c]{0.18\linewidth}
\centering
    \includegraphics[width=\textwidth]{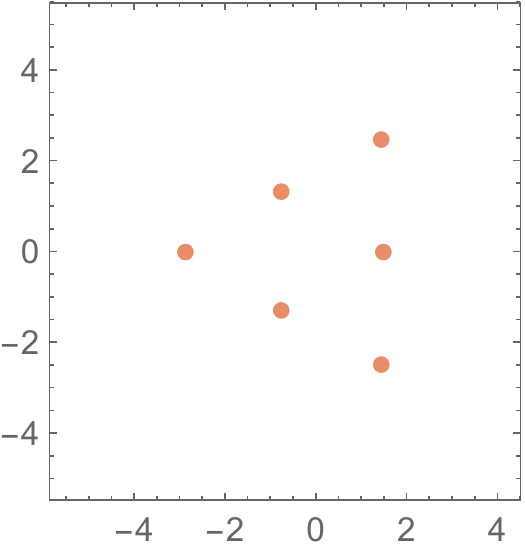}
    \text{\small $m = 3$}
\end{minipage}
\begin{minipage}[c]{0.18\linewidth}
\centering
    \includegraphics[width=\textwidth]{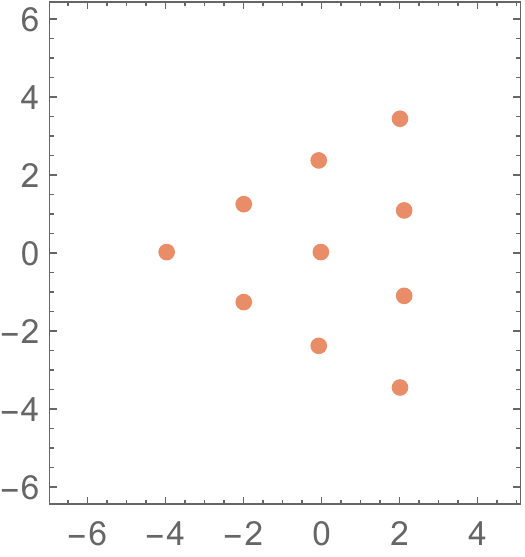}
    \text{\small $m = 4$}
\end{minipage}
\begin{minipage}[c]{0.18\linewidth}
\centering
    \includegraphics[width=\textwidth]{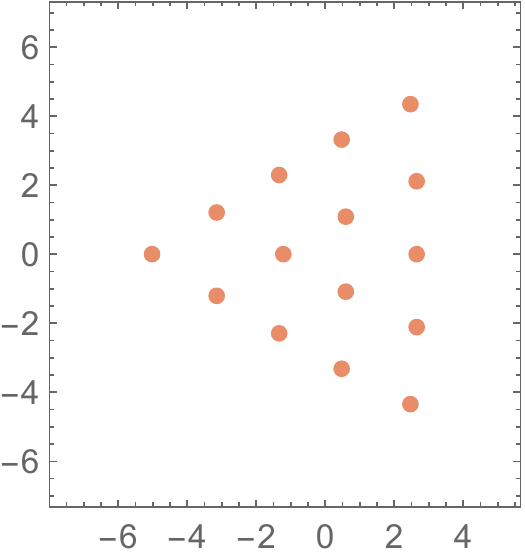}
    \text{\small $m = 5$}
\end{minipage}
\begin{minipage}[c]{0.18\linewidth}
\centering
    \includegraphics[width=\textwidth]{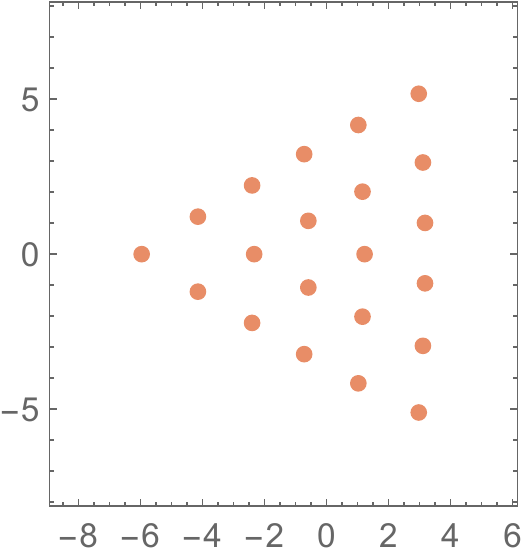}
    \text{\small $m = 6$}
\end{minipage}
\caption{\small Plots of Yablonskii-Vorobiev polynomials roots $u_m^{(1)} \brackets{z} = 0$, $m = 2, 3, 4, 5, 6$.}
\label{fig:rootsYV1}
\end{figure}

\begin{figure}[H]
\centering
\begin{minipage}[c]{  0.18\linewidth}
\centering
    \includegraphics[width=\textwidth]{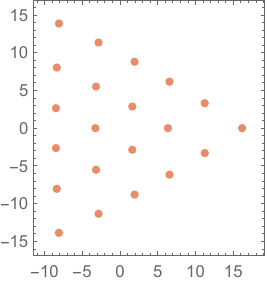}
    \text{\small $t_1 = -20$}
\end{minipage}
\begin{minipage}[c]{  0.18\linewidth}
\centering
    \includegraphics[width=\textwidth]{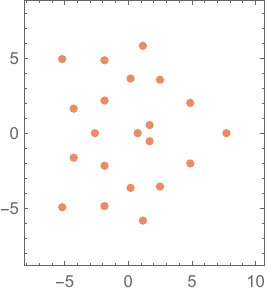}
    \text{\small $t_1 = -1$}
\end{minipage}
\begin{minipage}[c]{  0.18\linewidth}
\centering
    \includegraphics[width=\textwidth]{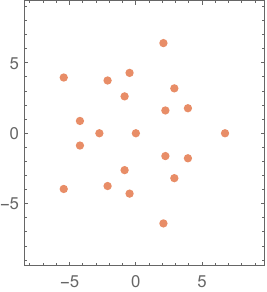}
    \text{\small $t_1 = 0$}
\end{minipage}
\begin{minipage}[c]{  0.18\linewidth}
\centering
    \includegraphics[width=\textwidth]{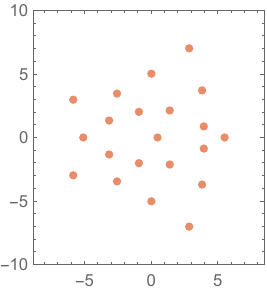}
    \text{\small $t_1 = 1$}
\end{minipage}
\begin{minipage}[c]{  0.18\linewidth}
\centering
    \includegraphics[width=\textwidth]{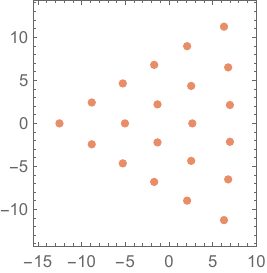}
    \text{\small $t_1 = 10$}
\end{minipage}
\caption{\small Plots of Yablonskii-Vorobiev-type polynomials roots $u_6^{(2)} \brackets{z, t_1} = 0$ at different $t_1$.}
\label{fig:rootsYV26}
\end{figure}

\begin{figure}[H]
\centering
\begin{minipage}[c]{  0.18\linewidth}
\centering
    \includegraphics[width=\textwidth]{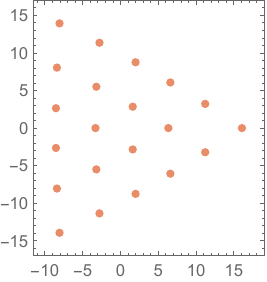}
    \text{\small $t_1 = -20$}
\end{minipage}
\begin{minipage}[c]{  0.18\linewidth}
\centering
    \includegraphics[width=\textwidth]{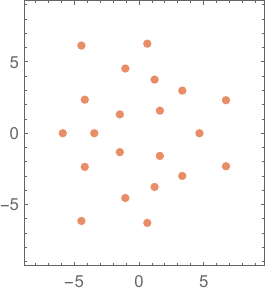}
    \text{\small $t_1 = -1$}
\end{minipage}
\begin{minipage}[c]{  0.18\linewidth}
\centering
    \includegraphics[width=\textwidth]{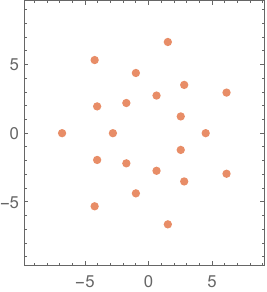}
    \text{\small $t_1 = 0$}
\end{minipage}
\begin{minipage}[c]{  0.18\linewidth}
\centering
    \includegraphics[width=\textwidth]{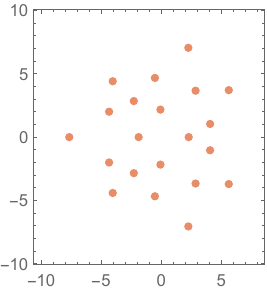}
    \text{\small $t_1 = 1$}
\end{minipage}
\begin{minipage}[c]{  0.18\linewidth}
\centering
    \includegraphics[width=\textwidth]{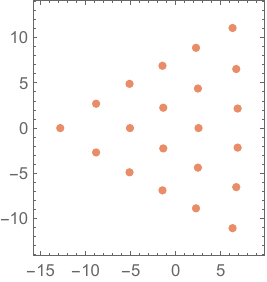}
    \text{\small $t_1 = 10$}
\end{minipage}
\caption{\small Plots of Yablonskii-Vorobiev-type polynomials roots $u_6^{(3)} \brackets{z, t_1, t_2} = 0$ at different $t_1$ and $t_2 = 0$.}
\label{fig:rootsYV63_1}
\end{figure}

\begin{figure}[H]
\centering
\begin{minipage}[c]{  0.18\linewidth}
\centering
    \includegraphics[width=\textwidth]{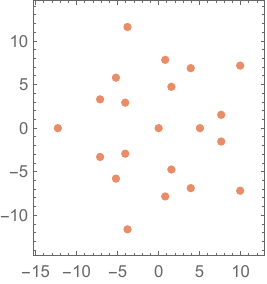}
    \text{\small $t_2 = -20$}
\end{minipage}
\begin{minipage}[c]{  0.18\linewidth}
\centering
    \includegraphics[width=\textwidth]{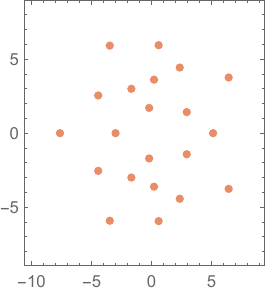}
    \text{\small $t_2 = -1$}
\end{minipage}
\begin{minipage}[c]{  0.18\linewidth}
\centering
    \includegraphics[width=\textwidth]{5_ap1/pictures/YBpol63_0_0}
    \text{\small $t_2 = 0$}
\end{minipage}
\begin{minipage}[c]{  0.18\linewidth}
\centering
    \includegraphics[width=\textwidth]{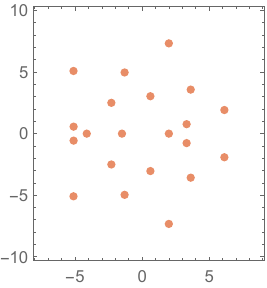}
    \text{\small $t_2 = 1$}
\end{minipage}
\begin{minipage}[c]{  0.18\linewidth}
\centering
    \includegraphics[width=\textwidth]{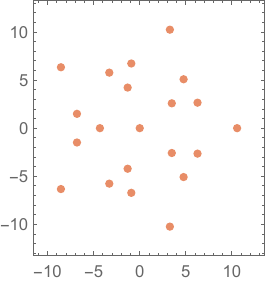}
    \text{\small $t_2 = 10$}
\end{minipage}
\caption{\small Plots of Yablonskii-Vorobiev-type polynomials roots $u_6^{(3)} \brackets{z, 0, t_2} = 0$ at $t_1 = 0$ and different $t_2$.}
\label{fig:rootsYV63_2}
\end{figure}

    \bibliographystyle{alpha}
    \bibliography{bib}
    \nocite{*}
\end{document}